%% file: jbhi_main.tex
\documentclass[journal,twoside,web]{ieeecolor}
\usepackage{generic}
\usepackage{cite}
\usepackage{amssymb,amsfonts}
\usepackage{algorithmic}
\usepackage{graphicx}
\usepackage{algorithm,algorithmic}
\usepackage{hyperref}
\usepackage{cite}

\usepackage{balance}

\makeatletter
\let\proof\@undefined
\let\endproof\@undefined
\makeatother
\usepackage{amsthm}

\newtheorem{theorem}{Theorem}

\theoremstyle{break}
\newtheorem{proposition}{Proposition}

\newtheorem{remark}{Remark}
\hypersetup{hidelinks=true}
\usepackage{pifont}
\newcommand{\cmark}{\ding{51}}
\newcommand{\xmark}{\ding{55}}

\usepackage{textcomp}
\usepackage{booktabs}   % \toprule \midrule \bottomrule
\usepackage{multirow}   % \multirow
\usepackage{amsmath}    % 数学符号（可选但建议）
\def\BibTeX{{\rm B\kern-.05em{\sc i\kern-.025em b}\kern-.08em
    T\kern-.1667em\lower.7ex\hbox{E}\kern-.125emX}}
\newcommand{\MedMamba}{\textbf{\texttt{MedMamba}}}

\title{\MedMamba: Recasting Mamba for Medical Time Series Classification}

\begin{document}

\author{
Zhengxiao He,
Huayu Li,
Xiwen Chen,
Janet M. Roveda,
Jinghao Wen,
Siyuan Tian
and Ao Li
\thanks{Z. He, H. Li, JM. Roveda, and A. Li are with the Department of Electrical \& Computer Engineering at the University of Arizona, Tucson, AZ, USA.}
\thanks{J. Wen is with the Department of Electrical and Computer Engineering, Villanova University, Villanova, PA, USA}
\thanks{S. Tian is with the Microsoft Research Asia, Beijing, China.}
\thanks{X. Chen is with Morgan Stanley, NY, USA
}
\thanks{
Corresponding author: Ao Li (e-mail: aoli1@arizona.edu).
}
}

\maketitle

\begin{abstract}
Medical time series, such as electrocardiograms (ECG) and electroencephalograms (EEG), exhibit complex temporal dynamics and structured cross-channel dependencies, posing fundamental challenges for automated analysis. Conventional convolutional and recurrent models struggle to capture long-range dependencies, while Transformer-based approaches incur quadratic complexity and often introduce redundant interactions that are misaligned with the intrinsic structure of physiological signals. To address these limitations, we propose \MedMamba{}, a principle-driven multi-scale bidirectional state space architecture tailored for medical time series classification. Our design is guided by three key inductive biases of physiological signals: spatial centralization, multi-timescale temporal composition, and non-causal contextual dependency. These principles are instantiated through a lightweight channel-mixing module for cross-channel reparameterization, multi-scale convolutional tokenization for temporal decomposition, and bidirectional Mamba blocks for efficient global context modeling with linear complexity. Extensive experiments on six benchmark datasets spanning EEG, ECG, and human activity signals demonstrate that \MedMamba{} consistently outperforms state-of-the-art methods across diverse modalities. Notably, it achieves 85.97\% accuracy on PTB and establishes new state-of-the-art performance on the challenging ADFTD dataset (54.72\% accuracy and 52.01\% F1-score). Strong results on long-sequence benchmarks, such as SleepEDF, further validate its capability in modeling long-range dependencies. Moreover, MedMamba achieves a speedup of 4.6$\times$ in inference, highlighting its practicality for real-time clinical deployment. These results suggest that principle-guided state space modeling offers an effective and scalable alternative to Transformer-based approaches for medical time series analysis.
\end{abstract}


\begin{figure}
    \centering
    \includegraphics[width=1.0\linewidth]{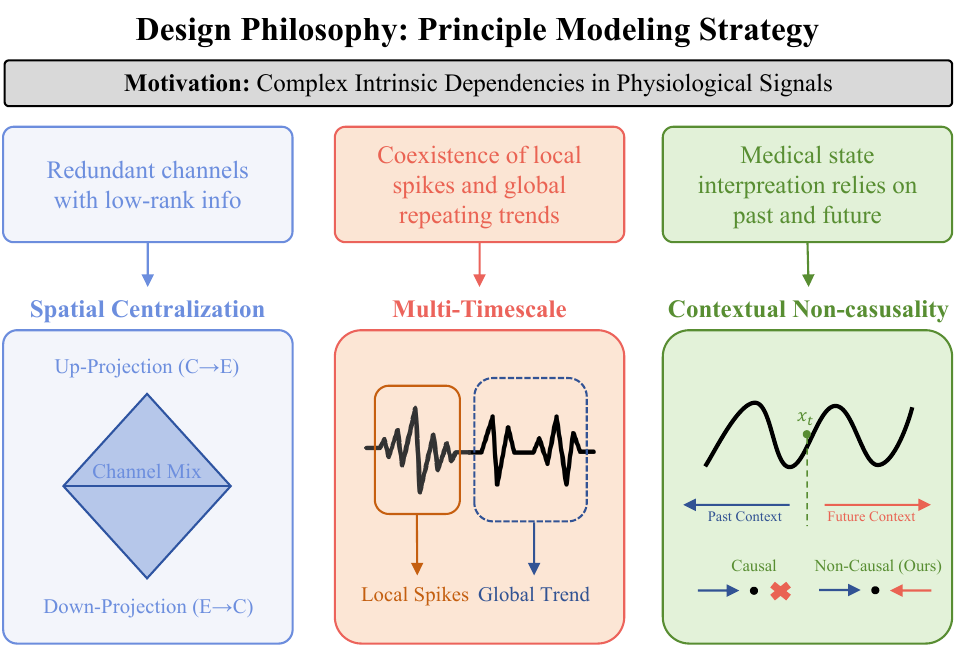}
   \caption{
        \textbf{Principle-driven framework for medical time series modeling.}
        \textbf{Top:} Medical signals inherently exhibit temporal heterogeneity and spatial dependency across channels, necessitating models that jointly capture both dimensions. 
        \textbf{Bottom:} Guided by these principles, we design a structured pipeline: 
        channel mixing for spatial reparameterization, 
        multi-scale embedding for temporal decomposition, 
        and bidirectional Mamba blocks for efficient non-causal context modeling. 
        These components collectively enable comprehensive representation learning across temporal scales and channel interactions.
        }
    \label{fig:teaser}
\end{figure}

\section{Introduction}
\label{sec:introduction}
\vspace{-5pt}
Physiological signals such as electrocardiograms (ECG) and electroencephalograms (EEG) are fundamental diagnostic tools in modern clinical practice, enabling the detection of cardiovascular and neurological disorders~\cite{ribeiro2020automatic,tang2025comprehensive,craik2019deep,lawhern2018eegnet,he2024eeg}. 
However, their interpretation relies heavily on expert knowledge, making the process time-consuming, subjective, and difficult to scale~\cite{hannun2019cardiologist}. 
With the rapid growth of continuous monitoring in clinical and wearable settings, there is an increasing need for automated and efficient analysis of medical time series.

Unlike generic sequential data, physiological signals exhibit distinct structural properties, including \textit{centralized spatial dependency}, \textit{multi-timescale temporal dynamics}, and \textit{non-causal contextual relationships}. 
These characteristics imply that diagnostically relevant information is concentrated in a subset of channels, distributed across multiple temporal scales, and often depends on both past and future context. 
However, most existing sequence modeling approaches are not designed to explicitly account for these properties.

Convolutional neural networks (CNNs)~\cite{acharya2017deep,kiranyaz2015real} effectively capture local morphological patterns but are limited by their receptive field. 
Recurrent neural networks (RNNs) can model temporal dependencies but suffer from vanishing gradients and limited parallelization~\cite{lipton2016learning}. 
More recently, Transformer-based models leverage self-attention to capture global context~\cite{wang2024medformer,zhou2021informer,liu2023itransformer,wu2021autoformer}, but incur quadratic complexity and introduce redundant interactions that are mismatched with the structured nature of physiological signals. 
As a result, existing methods either sacrifice long-range modeling or suffer from high computational cost.

State space models (SSMs) have recently emerged as an efficient alternative for long-sequence modeling~\cite{gu2021efficiently}. 
In particular, Mamba~\cite{gu2024mamba} introduces input-dependent selective scanning, enabling adaptive sequence modeling with linear complexity. 
Despite these advantages, existing Mamba-based approaches largely treat medical time series as generic sequences and do not explicitly incorporate their domain-specific structure, limiting their effectiveness in clinical applications.

To address these limitations, we propose \textbf{\MedMamba}, a principle-driven state space architecture tailored for medical time series classification. 
Our design is derived from three key inductive biases of physiological signals: (i) spatial centralization, (ii) multi-timescale temporal composition, and (iii) non-causal dependency. 
These principles are instantiated through a channel mixing module for cross-channel reparameterization, multi-scale convolutional tokenization for temporal decomposition, and bidirectional Mamba blocks for efficient global context modeling.

Extensive experiments on six benchmark datasets spanning EEG, ECG, and human activity signals demonstrate that \MedMamba{} consistently outperforms state-of-the-art methods across diverse modalities, while achieving linear computational complexity and improved inference efficiency. This work makes the following \textbf{contributions}:
\begin{enumerate}
\item We characterize the structural properties of medical time series, highlighting their centralized, multi-scale, and non-causal nature.
\item We propose \textbf{\MedMamba}, which, to the best of our knowledge, is the first principle-driven state space architecture integrating channel mixing, multi-scale modeling, and bidirectional inference.
\item We achieve efficient and scalable long-sequence modeling with linear complexity.
\item We demonstrate consistent performance gains across multiple medical benchmarks.
\end{enumerate}

\section{Related Works}
\label{sec:related}

\subsection{Deep Learning for Medical Time Series Classification}

Deep learning has become the dominant paradigm for medical time series classification~\cite{sun2020review}. 
Early approaches are primarily based on convolutional neural networks (CNNs), which effectively capture local morphological patterns in physiological signals, as demonstrated in ECG analysis~\cite{acharya2017deep,kiranyaz2015real} and EEG modeling~\cite{lawhern2018eegnet}. 
However, their limited receptive field restricts the modeling of long-range temporal dependencies, which are critical in applications such as sleep staging and disease progression.

To address this limitation, recurrent neural networks (RNNs), particularly long short-term memory (LSTM) models, were introduced to model temporal dependencies~\cite{lipton2016learning}. 
While capable of capturing sequential information, they suffer from vanishing gradients and limited parallelization, making them less suitable for long sequences.

More recently, Transformer-based models have been applied to medical time series, leveraging attention mechanisms to model global dependencies. 
Variants such as Informer~\cite{zhou2021informer}, Autoformer~\cite{wu2021autoformer},  iTransformer~\cite{liu2023itransformer} and Medformer~\cite{wang2024medformer} have shown strong performance. 
However, their quadratic complexity and dense token interactions introduce substantial computational overhead and redundancy for physiological signals.

\subsection{State Space Models for Time Series}

State space models (SSMs)~\cite{gu2021efficiently} have recently gained attention for their efficiency and ability to capture long-range dependencies. 
Mamba~\cite{gu2024mamba} further improves SSMs through input-dependent selective scanning, enabling adaptive sequence modeling.

Although Mamba-based methods have achieved promising results in general time series tasks~\cite{wang2025mamba,liu2026gcmnet}, their application to medical time series remains limited. 
Existing works such as ECG-Mamba~\cite{jiang2025ecg} largely treat physiological signals as generic sequences without considering their domain-specific properties.

In particular, medical time series exhibit multi-scale temporal dynamics and structured physiological patterns, which are not explicitly modeled in existing Mamba-based approaches, limiting their effectiveness in clinical settings.

\begin{figure*}[t]
    \centering
    \includegraphics[width=\textwidth]{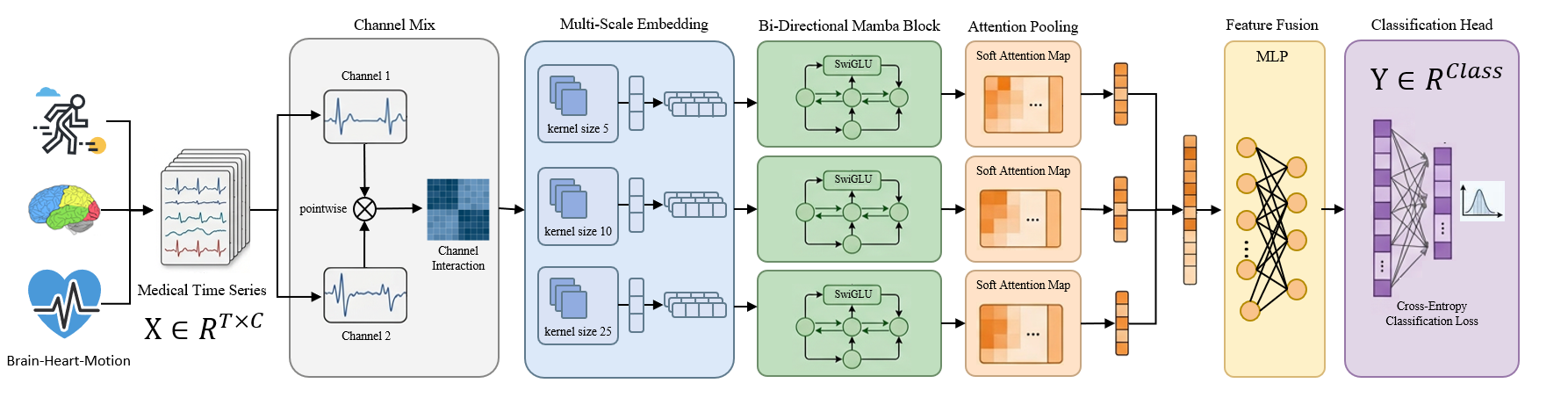}
   \caption{
Overview of \textbf{\MedMamba }, a multi-scale bidirectional Mamba architecture for medical time series analysis.
The framework first performs channel-aware preprocessing to capture cross-channel interactions in physiological signals.
It then employs multi-scale convolutional embeddings to encode temporal dynamics at different resolutions, enabling simultaneous modeling of fast local variations and slow global trends.
Bidirectional Mamba blocks are applied at each scale to efficiently capture long-range dependencies with linear complexity.
Finally, an attention-based fusion module integrates multi-scale representations for robust classification.
}\label{fig:model_teaser}
\end{figure*}

\section{Method}
\label{sec:method}
This section presents \MedMamba{} from a principle-driven perspective. Rather than describing the model as a flat sequence of engineering modules, we start with the structural properties of physiological signals and then show how each architectural component is derived from these properties. The resulting presentation follows the logic
\emph{data structure, inductive bias, and architectural instantiation}.

% ================================================================
\subsection{Design Principles}
\label{sec:design_principles}

Medical time series are not generic token sequences. Their statistical structure suggests three inductive biases that should be explicitly reflected in the architecture.

\textbf{P1: Centralized spatial dependency.} In multichannel physiological recordings, diagnostically relevant variation is often concentrated in a relatively small spatial subspace. In practice, this means that a subset of channels or latent physiological sources carries disproportionate explanatory power, while the remaining channels are partially redundant or noisy. A suitable encoder should therefore reparameterize the channel space before temporal modeling instead of deferring all cross-channel interaction to the end.

\textbf{P2: Multi-timescale temporal composition.} Clinical evidence is distributed across distinct temporal scales. Short-lived events such as QRS complexes, epileptic spikes, or movement bursts coexist with slower rhythms such as heart-rate variability and sleep-stage transitions. A single tokenization scale therefore induces an avoidable resolution mismatch between the model and the signal.

\textbf{P3: Non-causal contextual dependency.} Many clinical patterns are defined not only by what happens before a time point, but also by what happens after it. This is especially true for symmetric, phase-dependent, or boundary-sensitive waveforms. For offline classification, a strictly causal temporal encoder is thus structurally suboptimal.

Taken together, these observations motivate a \emph{channel-first, multi-rate, bidirectional state-space architecture}: the model should first refine the channel basis, then expose the signal at multiple temporal resolutions, and finally propagate information in both temporal directions using linear-time state-space dynamics.

% ================================================================
\subsection{Problem Formulation}
\label{sec:formulation}

Given a multi-channel physiological recording $\mathbf{X} \in \mathbb{R}^{L \times C}$, where $L$ denotes the number of time steps and $C$ denotes the number of channels, the goal is to learn a mapping
\begin{equation}
    f: \mathbb{R}^{L \times C} \rightarrow \{1,2,\ldots,K\}
\end{equation}
that assigns the input to one of $K$ diagnostic categories.

The task contains two coupled challenges. First, the encoder must model \emph{temporal heterogeneity}: both short-range morphologies and long-range dependencies must be preserved. Second, it must model \emph{spatial dependency}: the predictive signal is not factorized across channels because ECG leads, EEG electrodes, and IMU axes are physiologically coupled. \MedMamba{} addresses these challenges by factorizing the representation pipeline into three stages:
To address the intrinsic challenges of medical time series, including spatial dependency, temporal heterogeneity, and non-causal context, we factorize the representation learning process into a sequence of principle-driven stages. Given the input $\mathbf{X}$, we first apply robust preprocessing to obtain $\tilde{\mathbf{X}}$, followed by spatial reparameterization to produce $\mathbf{X}'$. The reparameterized signal is then encoded through multi-timescale temporal modeling, resulting in a set of scale-specific representations $\{\mathbf{H}^{(m)}\}_{m=1}^{M}$. 
Finally, these representations are aggregated across scales to produce the prediction $\hat{\mathbf{y}}$. 
Specifically, the multi-scale features $\{\mathbf{H}^{(m)}\}_{m=1}^{M}$ are concatenated along the feature dimension via $\bigoplus$, forming a unified representation that captures both short- and long-term temporal patterns. 
The aggregation function $\phi(\cdot)$ further performs feature transformation and dimensionality alignment, enabling effective cross-scale interaction. 
The final prediction is obtained through a linear classifier:
\begin{equation}
\hat{y} = \arg\max_{k}\;
\mathbf{W}\,\phi\!\left(
\bigoplus_{m=1}^{M} \mathbf{H}^{(m)}
\right) + \mathbf{b},
\end{equation}
where $\mathbf{W}$ and $\mathbf{b}$ denote the learnable parameters of the classifier.
% ================================================================
\subsection{From Principles to Architecture}
\label{sec:overview}

The overall design of MedMamba is illustrated in Fig.~\ref{fig:model_teaser}. For clarity, the three principles introduced above are instantiated through five computational stages. In the \textbf{first stage}, robust input modeling is achieved through channel dropout, which regularizes the model against sensor corruption or missing leads while preserving temporal coherence within each channel. In the \textbf{second stage}, spatial reparameterization is performed by a lightweight channel-mixing block operating along the channel axis at each time step, thereby reducing redundancy and aligning the representation with dominant physiological subspaces. The \textbf{third stage} introduces multi-rate tokenization, where parallel convolutional stems with different strides expose the signal on multiple temporal lattices and mitigate the resolution mismatch between token scale and event duration. In the \textbf{fourth stage}, each scale is processed by bidirectional Mamba blocks that integrate local convolutional refinement, selective state-space scanning, and gated residual updates. Finally, in the \textbf{fifth stage}, cross-scale evidence aggregation is carried out through attention pooling, which compresses each scale into a fixed-length descriptor before these descriptors are fused for classification. This decomposition makes the information flow explicit: \MedMamba{} first reduces spatial redundancy, then constructs scale-specific temporal summaries, and ultimately aggregates diagnostic evidence across scales.

% ================================================================
\subsection{Robust Input Modeling via Channel Dropout}
\label{sec:preprocessing}

To improve robustness to missing or corrupted sensors, we apply channel-level dropout before feature extraction. Each channel is independently dropped with probability $p_{\mathrm{ch}}$, and the remaining channels are rescaled:
\begin{equation}
    \tilde{\mathbf{X}}_{:,c} =
    \begin{cases}
        \dfrac{\hat{\mathbf{X}}_{:,c}}{1-p_{\mathrm{ch}}}, & \text{if } m_c=1, \\
        \mathbf{0}, & \text{if } m_c=0,
    \end{cases}
\end{equation}
where $m_c \sim \mathrm{Bernoulli}(1-p_{\mathrm{ch}})$ is shared across time steps.

This design preserves channel-wise temporal structure while improving robustness to sensor-level perturbations.

% ================================================================
% ================================================================
\subsection{Principle I: Spatial Centralization via Channel Mixing}
\label{sec:channel_mix}

After robust preprocessing, \MedMamba{} applies a channel-mixing layer to model the cross-channel structure before temporal encoding. For each time step, the layer performs a two-layer MLP with residual connection:
\begin{equation}
    \mathbf{X}' = \tilde{\mathbf{X}} + \mathrm{MLP}\!\left(\mathrm{LN}(\tilde{\mathbf{X}})\right),
\end{equation}
where
\begin{equation}
\mathrm{MLP}(\mathbf{h}) = \mathbf{W}_2\,\mathrm{Dropout}\!\left(\mathrm{GELU}(\mathbf{W}_1\mathbf{h})\right),
\end{equation}
with
\begin{equation}
\mathbf{W}_1 \in \mathbb{R}^{\rho C' \times C'}, \qquad
\mathbf{W}_2 \in \mathbb{R}^{C' \times \rho C'}.
\end{equation}
Here $\rho$ is the expansion ratio and LN denotes LayerNorm.

The role of this block is not merely to mix features. It implements a \emph{spatial reparameterization}: the model learns a better channel basis before any long-range temporal reasoning is performed. This is advantageous when the original channel coordinates are overspecified relative to the intrinsic dimensionality of the underlying physiology.

\begin{proposition}[Noise suppression under centralized channel structure]
\label{prop:noise_suppression}
Assume that at a fixed time step the channel vector admits the decomposition
\begin{equation}
    \mathbf{x}_t = \mathbf{s}_t + \mathbf{n}_t,
\end{equation}
where $\mathbf{s}_t \in \mathcal{U}$ lies in an informative subspace $\mathcal{U} \subseteq \mathbb{R}^{C}$ and $\mathbf{n}_t \in \mathcal{U}^{\perp}$ is nuisance energy. Let $\mathbf{M}_t$ denote the local linearization of the channel-mixing map at $\mathbf{x}_t$. If $\mathbf{M}_t$ approximately preserves $\mathcal{U}$ while contracting $\mathcal{U}^{\perp}$, namely
\begin{equation}
    \|\mathbf{M}_t\mathbf{s}_t\|_2 \geq (1-\epsilon)\|\mathbf{s}_t\|_2,
    \qquad
    \|\mathbf{M}_t\mathbf{n}_t\|_2 \leq \gamma \|\mathbf{n}_t\|_2
\end{equation}
for some $\epsilon \in [0,1)$ and $\gamma \in [0,1)$, then the nuisance-to-signal ratio is reduced after mixing:
\begin{equation}
    \frac{\|\mathbf{M}_t\mathbf{n}_t\|_2}{\|\mathbf{M}_t\mathbf{s}_t\|_2}
    \leq
    \frac{\gamma}{1-\epsilon}
    \frac{\|\mathbf{n}_t\|_2}{\|\mathbf{s}_t\|_2}.
\end{equation}
Whenever $\gamma < 1-\epsilon$, the mixed representation is strictly more concentrated on the informative subspace.
\end{proposition}

\begin{proof}
The result follows directly from the two norm inequalities. The numerator contracts by at most $\gamma$, whereas the denominator contracts by at most $(1-\epsilon)$. Dividing the two bounds yields the stated inequality.
\end{proof}

The proposition above characterizes channel mixing from an energy perspective. We now complement it with an information-theoretic interpretation. Let $\mathbf{z}_t = g(\mathbf{x}_t)$ denote the output of the channel-mixing block at time $t$, and let $y$ be the diagnostic label. Since $\mathbf{z}_{1:L}$ is a deterministic transformation of $\mathbf{x}_{1:L}$, the data processing inequality implies
\begin{equation}
    I(\mathbf{z}_{1:L}; y) \leq I(\mathbf{x}_{1:L}; y).
\end{equation}
Hence, channel mixing does not create new label information. Its role is instead to \emph{reorganize} the channel space so that task-relevant information is concentrated into a lower-redundancy representation that is easier for the downstream temporal encoder to exploit under finite model capacity.

To make this precise, we analyze the channel mixer locally around a working point $\mathbf{x}_t^\star$. Let
\begin{equation}
    g(\mathbf{x}_t) \approx g(\mathbf{x}_t^\star) + \mathbf{M}_t(\mathbf{x}_t-\mathbf{x}_t^\star),
\end{equation}
where $\mathbf{M}_t = J_g(\mathbf{x}_t^\star)$ is the Jacobian of the nonlinear channel-mixing map. Ignoring the constant offset, we obtain the local model
\begin{equation}
    \mathbf{z}_t \approx \mathbf{M}_t\mathbf{s}_t + \mathbf{M}_t\mathbf{n}_t + \boldsymbol{\xi}_t,
\end{equation}
where $\boldsymbol{\xi}_t \sim \mathcal{N}(\mathbf{0},\sigma^2\mathbf{I})$ is an auxiliary readout noise introduced for analysis. Assume further that
\begin{equation}
    \mathbf{s}_t \sim \mathcal{N}(\mathbf{0},\mathbf{\Sigma}_s), 
    \qquad
    \mathbf{n}_t \sim \mathcal{N}(\mathbf{0},\mathbf{\Sigma}_n),
    \qquad
    \mathbf{s}_t \perp \mathbf{n}_t.
\end{equation}
Then the mutual information between the mixed representation and the informative component is
\begin{equation}
    I(\mathbf{z}_t;\mathbf{s}_t)
    =
    \frac{1}{2}
    \log
    \frac{
    \det\!\big(\mathbf{M}_t(\mathbf{\Sigma}_s+\mathbf{\Sigma}_n)\mathbf{M}_t^\top+\sigma^2\mathbf{I}\big)
    }{
    \det\!\big(\mathbf{M}_t\mathbf{\Sigma}_n\mathbf{M}_t^\top+\sigma^2\mathbf{I}\big)
    }.
\end{equation}
Equivalently,
\begin{equation}
    I(\mathbf{z}_t;\mathbf{s}_t)
    =
    \frac{1}{2}
    \log
    \det\!\left[
        \mathbf{I}
        +
        \big(\mathbf{M}_t\mathbf{\Sigma}_n\mathbf{M}_t^\top+\sigma^2\mathbf{I}\big)^{-1}
        \mathbf{M}_t\mathbf{\Sigma}_s\mathbf{M}_t^\top
    \right].
\end{equation}
This expression shows that the effectiveness of channel mixing is governed by the \emph{generalized signal-to-noise ratio} after projection: a good mixer is one that preserves directions carrying label-relevant physiological variation while suppressing nuisance covariance.

\begin{proposition}[Optimal channel mixing under generalized signal-to-noise ratio]
Consider the class of rank-$r$ local linear mixers $\mathbf{M}$ satisfying the whitening constraint
\begin{equation}
    \mathbf{M}(\mathbf{\Sigma}_n+\sigma^2\mathbf{I})\mathbf{M}^\top=\mathbf{I}_r.
\end{equation}
Define the whitened signal covariance
\begin{equation}
    \mathbf{K}
    =
    (\mathbf{\Sigma}_n+\sigma^2\mathbf{I})^{-1/2}
    \mathbf{\Sigma}_s
    (\mathbf{\Sigma}_n+\sigma^2\mathbf{I})^{-1/2}.
\end{equation}
Then the mixer maximizing $I(\mathbf{z}_t;\mathbf{s}_t)$ is obtained by projecting onto the top-$r$ eigendirections of $\mathbf{K}$. In particular, the optimal value is
\begin{equation}
    \max_{\mathbf{M}} I(\mathbf{z}_t;\mathbf{s}_t)
    =
    \frac{1}{2}\sum_{i=1}^{r}\log(1+\lambda_i(\mathbf{K})),
\end{equation}
where $\lambda_1(\mathbf{K}) \geq \cdots \geq \lambda_r(\mathbf{K})$ are the largest eigenvalues of $\mathbf{K}$.
\end{proposition}

\begin{proof}
Under the whitening constraint, any feasible $\mathbf{M}$ can be written as
\begin{equation}
    \mathbf{M}
    =
    \mathbf{U}^\top(\mathbf{\Sigma}_n+\sigma^2\mathbf{I})^{-1/2},
    \qquad
    \mathbf{U}^\top\mathbf{U}=\mathbf{I}_r.
\end{equation}
Substituting this form into the mutual-information objective yields
\begin{equation}
    I(\mathbf{z}_t;\mathbf{s}_t)
    =
    \frac{1}{2}
    \log\det\!\left(\mathbf{I}_r+\mathbf{U}^\top\mathbf{K}\mathbf{U}\right).
\end{equation}
By the Ky Fan maximum principle, the quantity $\log\det(\mathbf{I}_r+\mathbf{U}^\top\mathbf{K}\mathbf{U})$ is maximized when the columns of $\mathbf{U}$ span the principal $r$-dimensional eigensubspace of $\mathbf{K}$. Substituting the optimal $\mathbf{U}$ gives
\begin{equation}
    \max_{\mathbf{M}} I(\mathbf{z}_t;\mathbf{s}_t)
    =
    \frac{1}{2}\sum_{i=1}^{r}\log(1+\lambda_i(\mathbf{K})).
\end{equation}
This proves the claim.
\end{proof}

The proposition yields a concrete interpretation of spatial centralization. When the informative covariance is concentrated in a small number of dominant generalized eigendirections, most task-relevant information can be preserved in a low-dimensional mixed representation, while nuisance variation is discarded. In other words, the ideal channel mixer behaves like an adaptive, task-aware spatial filter that first whitens nuisance structure and then aligns the representation with the most informative physiological subspace.

This information-theoretic view is also consistent with Proposition \ref{prop:noise_suppression}. If the local Jacobian approximately preserves the informative subspace while contracting the nuisance subspace, then the effective post-mixing signal-to-noise ratio is increased, which in turn increases the mutual information between the mixed representation and the informative latent component. Thus, the energy-based and information-based perspectives describe the same phenomenon from complementary angles: channel mixing improves the concentration, separability, and usability of physiological information before long-range temporal modeling.

Although the implemented channel mixer is nonlinear, the above derivation formalizes its local effect: around a working point, the residual MLP behaves like an adaptive spatial filter that can preserve dominant physiological directions while attenuating channel-specific nuisance variation.

% ================================================================
\subsection{Principle II: Multi-Rate Temporal Decomposition}
\label{sec:patch_embed}

To model temporally heterogeneous evidence, we tokenize the signal using $M$ parallel convolutional stems with different kernel sizes and strides. Let $\mathcal{S}=\{s_1,s_2,\ldots,s_M\}$ denote the stride set. The $m$-th stem applies a 1-D convolution with kernel size $s_m$ and stride $s_m$, followed by batch normalization and GELU activation:
\begin{equation}
    \mathbf{P}^{(m)} = \mathrm{GELU}\!\left(\mathrm{BN}\!\left(\mathrm{Conv1D}_{s_m}(\mathbf{X}')\right)\right) \in \mathbb{R}^{L_m \times D},
\end{equation}
where
\begin{equation}
    L_m = \left\lfloor \frac{L-s_m}{s_m} \right\rfloor + 1
\end{equation}
is the number of tokens at scale $m$. A learnable positional embedding $\mathbf{E}^{(m)} \in \mathbb{R}^{L_m \times D}$ is then added:
\begin{equation}
    \mathbf{H}^{(m)}_0 = \mathbf{P}^{(m)} + \mathbf{E}^{(m)}.
\end{equation}

This construction can be interpreted as a \emph{learnable multi-rate analysis bank}. Small strides retain high-resolution morphology, whereas large strides compress the signal and emphasize coarse temporal organization. The purpose is not only better feature diversity, but also better alignment between token scale and event duration.

\begin{proposition}[Worst-case scale mismatch under multi-scale tokenization]
Let $s_1 < s_2 < \cdots < s_M$ and define the log-scale mismatch of an event with temporal support $\tau \in [s_1,s_M]$ by
\begin{equation}
    d(\tau;\mathcal{S}) = \min_{1 \leq m \leq M} \left|\log \tau - \log s_m\right|.
\end{equation}
Then the worst-case mismatch satisfies
\begin{equation}
    \max_{\tau \in [s_1,s_M]} d(\tau;\mathcal{S})
    = \frac{1}{2}\max_{1 \leq m < M} \log\!\left(\frac{s_{m+1}}{s_m}\right).
\end{equation}
Therefore, denser scale coverage on the log-time axis strictly improves the worst-case alignment between event duration and token scale.
\end{proposition}
\begin{proof}
On the log-time axis, the candidate scales are the points $\{\log s_m\}_{m=1}^{M}$. The farthest event from its nearest scale lies at the midpoint of the largest adjacent interval, hence the worst-case nearest-neighbor distance equals half the largest gap. Exponentiating the gap yields the stated expression.
\end{proof}

The proposition provides a concrete interpretation of why multi-scale tokenization is preferable to any single fixed stride: it reduces the maximum representation mismatch over a range of clinically relevant temporal supports.

\begin{theorem}[Linear-time complexity under multi-scale bidirectional scanning]
Let $D_{\mathrm{inner}} = \alpha D$ and let $N$ denote the SSM state dimension. For stride set $\mathcal{S} = \{s_m\}_{m=1}^{M}$ with patch lengths $L_m = \lfloor (L-s_m)/s_m \rfloor + 1$, the total complexity of the bidirectional selective scans across all scales is
\begin{equation}
    \mathcal{O}\!\left(\sum_{m=1}^{M} L_m D_{\mathrm{inner}} N\right).
\end{equation}
Moreover, since $L_m \leq L/s_m + 1$,
\begin{equation}
    \mathcal{O}\!\left(\sum_{m=1}^{M} L_m D_{\mathrm{inner}} N\right)
    = \mathcal{O}\!\left(L D_{\mathrm{inner}} N \sum_{m=1}^{M} \frac{1}{s_m}\right),
\end{equation}
which is linear in the original sequence length $L$.
\end{theorem}
\begin{proof}
For a branch of length $L_m$, one selective scan has complexity $\mathcal{O}(L_m D_{\mathrm{inner}} N)$. Bidirectional scanning introduces only a constant factor of two, which is absorbed into big-$\mathcal{O}$ notation. Summing over branches gives the first expression. Substituting the bound $L_m \leq L/s_m + 1$ yields the second expression.
\end{proof}

This theorem explains the central efficiency advantage of \MedMamba{}: the model does not pay quadratic cost for global temporal interaction, and the multi-scale stems further shorten the sequences processed by the state-space backbone.

% ================================================================
\subsection{Principle III: Bidirectional State-Space Inference}
\label{sec:bimamba}

\subsubsection{State Space Model Preliminaries}
State space models (SSMs) map a one-dimensional input signal $u(t) \in \mathbb{R}$ to an output $y(t) \in \mathbb{R}$ through a latent state $\mathbf{h}(t) \in \mathbb{R}^{N}$ via the continuous-time system
\begin{equation}
    \mathbf{h}'(t) = \mathbf{A}\mathbf{h}(t) + \mathbf{B}u(t), \qquad
    y(t) = \mathbf{C}\mathbf{h}(t) + Du(t),
    \label{eq:ssm_cont}
\end{equation}
where $\mathbf{A} \in \mathbb{R}^{N \times N}$ is the state transition matrix, $\mathbf{B} \in \mathbb{R}^{N \times 1}$ and $\mathbf{C} \in \mathbb{R}^{1 \times N}$ are input/output projections, and $D \in \mathbb{R}$ is a feed-through scalar.

For discrete processing, Eq.~\eqref{eq:ssm_cont} is discretized using zero-order hold with step size $\Delta$:
\begin{equation}
    \bar{\mathbf{A}} = \exp(\Delta \mathbf{A}), \qquad
    \bar{\mathbf{B}} = (\Delta\mathbf{A})^{-1}(\exp(\Delta\mathbf{A}) - \mathbf{I})\,\Delta\mathbf{B},
\end{equation}
which yields the recurrence
\begin{equation}
    \mathbf{h}_t = \bar{\mathbf{A}}\mathbf{h}_{t-1} + \bar{\mathbf{B}}u_t, \qquad
    y_t = \mathbf{C}\mathbf{h}_t + Du_t.
    \label{eq:ssm_disc}
\end{equation}

This formulation is especially attractive for medical time series because it decouples \emph{state dimension} from \emph{sequence length}: long temporal context can be accumulated in the latent state without explicitly materializing pairwise interactions between all time points.

\subsubsection{Selective Scan Mechanism}
Following Mamba~\cite{gu2024mamba}, we make the SSM input-dependent so that the memory dynamics adapt to the signal content. For each position $t$ and feature dimension $d$, the step size and input/output projections are generated from the current feature vector $\mathbf{x}_t$:
\begin{align}
    \Delta_{t,d}
    &= \mathrm{softplus}\!\left(\mathbf{w}_{\Delta,d}^{\top}\mathbf{x}_t + b_{\Delta,d}\right),
    \label{eq:delta} \\
    \mathbf{B}_t &= \mathbf{W}_B\mathbf{x}_t,
    \label{eq:B} \\
    \mathbf{C}_t &= \mathbf{W}_C\mathbf{x}_t.
    \label{eq:C}
\end{align}
The transition matrix $\mathbf{A}$ is shared across time and parameterized in logarithmic space using S4D initialization~\cite{gu2021efficiently}. The resulting selective scan can be interpreted as an adaptive memory mechanism: salient waveform segments may trigger short update scales and stronger input injection, while smoother regions may favor slower latent evolution.

\subsubsection{Bidirectional Scanning}
To operationalize Principle III, \MedMamba{} applies selective scan in both forward and backward directions. Given the input $\mathbf{x}_{\mathrm{ssm}} \in \mathbb{R}^{L' \times D_{\mathrm{inner}}}$, we compute
\begin{equation}
\begin{aligned}
\mathbf{y}^{(k)}
= \mathrm{SelectiveScan}\bigl(
&\mathbf{x}_{\mathrm{ssm}}^{(k)},\,
\Delta^{(k)},\,
\mathbf{A}, \\
&\mathbf{B}^{(k)},\,
\mathbf{C}^{(k)},\,
D
\bigr), 
\end{aligned}
\end{equation}
where $k \in \{\mathrm{fwd}, \mathrm{bwd}\}$. The two directional states are then merged into a context-enriched representation $\mathbf{y}$.

\begin{proposition}[Necessity of bidirectional context for symmetric dependencies]
Suppose that the discriminative statistic at position $t$ depends on a symmetric neighborhood $\mathbf{x}_{t-r:t+r}$ for some $r>0$. Then any strictly causal representation of the form $\mathbf{z}_t = F(\mathbf{x}_{\leq t})$ cannot recover that statistic at time $t$ without an explicit delay of at least $r$ steps, whereas a bidirectional representation of the form $\mathbf{z}_t = G(\mathbf{x}_{\leq t}, \mathbf{x}_{\geq t})$ can access the full neighborhood at the same index.
\end{proposition}
\begin{proof}
A strictly causal encoder has no access to the future samples $\mathbf{x}_{t+1},\ldots,\mathbf{x}_{t+r}$ when constructing $\mathbf{z}_t$. Therefore, any statistic that depends essentially on the full symmetric window cannot be computed at position $t$ unless the representation is shifted by at least $r$ steps. A bidirectional encoder maintains both past and future-conditioned states, making the full neighborhood available when forming the representation at position $t$.
\end{proof}

The proposition is particularly relevant for medical classification, where the model is evaluated on complete segments rather than in a strictly streaming regime. In this setting, bidirectional inference supplies a better inductive bias than a causal-only encoder.

\begin{remark}
Under this interpretation, the bidirectional \MedMamba{} block behaves closer to a forward-backward smoothing operator than to a purely causal filter. This distinction matters because offline diagnosis is usually based on fully observed windows, not on online forecasting constraints.
\end{remark}

\subsubsection{Block Architecture}
\label{sec:bimamba_block}
Each bidirectional Mamba block processes a scale-specific sequence $\mathbf{H} \in \mathbb{R}^{L' \times D}$ through five steps.

\textbf{(1) Input projection and gating.}
The normalized input is projected into a state-space branch and a gating branch:
\begin{equation}
    [\mathbf{x}_{\mathrm{ssm}},\mathbf{g}] = \mathrm{Linear}(\mathrm{LN}(\mathbf{H})) \in \mathbb{R}^{L' \times 2D_{\mathrm{inner}}},
\end{equation}
where $D_{\mathrm{inner}}=\alpha D$.

\textbf{(2) Local refinement.}
A depthwise 1-D convolution injects short-range temporal structure:
\begin{equation}
    \mathbf{x}_{\mathrm{ssm}} \leftarrow \mathrm{SiLU}\!\left(\mathrm{DWConv1D}(\mathbf{x}_{\mathrm{ssm}})\right).
\end{equation}
This step complements the long-range state dynamics by providing a local morphological prior.

\textbf{(3) Bidirectional selective scan.}
The refined sequence is processed in both directions to produce $\mathbf{y}_{\mathrm{bi}} \in \mathbb{R}^{L' \times D_{\mathrm{inner}}}$.

\textbf{(4) Gated output projection.}
The scan output is modulated by the learned gate and projected back to the model dimension:
\begin{equation}
    \mathrm{out} = \mathrm{Linear}\!\left(\mathrm{LN}(\mathbf{y}_{\mathrm{bi}}) \odot \mathrm{SiLU}(\mathbf{g})\right) \in \mathbb{R}^{L' \times D}.
\end{equation}
The gate makes the residual update content-adaptive: the model may choose to pass through or suppress state information depending on the local evidence.

\textbf{(5) Residual update.}
Finally,
\begin{equation}
    \mathbf{H} \leftarrow \mathbf{H} + \mathrm{DropPath}(\mathrm{out}).
\end{equation}
The residual form stabilizes optimization and preserves a direct information path across deep stacks.

% ================================================================
\subsection{Gated Feed-Forward Refinement}
\label{sec:gffn}

Each bidirectional Mamba block is followed by a gated feed-forward network using the SwiGLU form:
\begin{equation}
    \mathrm{GatedFFN}(\mathbf{H}) = \mathbf{W}_{\mathrm{down}}\!\left(\mathrm{SiLU}(\mathbf{W}_{\mathrm{gate}}\mathbf{H}) \odot \mathbf{W}_{\mathrm{up}}\mathbf{H}\right),
\end{equation}
where $\mathbf{W}_{\mathrm{gate}}, \mathbf{W}_{\mathrm{up}} \in \mathbb{R}^{D_{\mathrm{ffn}} \times D}$ and $\mathbf{W}_{\mathrm{down}} \in \mathbb{R}^{D \times D_{\mathrm{ffn}}}$ with $D_{\mathrm{ffn}} = \beta D$. The block is wrapped with a pre-normalized residual connection:
\begin{equation}
    \mathbf{H} \leftarrow \mathbf{H} + \mathrm{DropPath}\!\left(\mathrm{GatedFFN}(\mathrm{LN}(\mathbf{H}))\right).
\end{equation}
This module provides a complementary source of nonlinearity after state propagation. In other words, the SSM is responsible for structured temporal transport, while the GatedFFN increases feature selectivity and channel-wise interaction in the hidden space.

% ================================================================
\subsection{Cross-Scale Evidence Aggregation}
\label{sec:pooling}

After $N_{\mathrm{layer}}$ stacked blocks, each scale branch produces contextualized tokens $\mathbf{H}^{(m)} \in \mathbb{R}^{L_m \times D}$. We summarize each branch by attention pooling:
\begin{equation}
    \alpha_t^{(m)} = \frac{\exp\!\left(\mathbf{w}_2^{\top}\tanh(\mathbf{W}_1\mathbf{h}_t^{(m)})\right)}{\sum_{t'=1}^{L_m} \exp\!\left(\mathbf{w}_2^{\top}\tanh(\mathbf{W}_1\mathbf{h}_{t'}^{(m)})\right)},
\end{equation}
\begin{equation}
    \mathbf{r}^{(m)} = \sum_{t=1}^{L_m} \alpha_t^{(m)}\mathbf{h}_t^{(m)} \in \mathbb{R}^{D},
\end{equation}
where $\mathbf{W}_1 \in \mathbb{R}^{(D/4) \times D}$ and $\mathbf{w}_2 \in \mathbb{R}^{D/4}$ are learnable parameters.

The scale descriptors are then fused:
\begin{equation}
\begin{aligned}
\mathbf{z}
= \mathrm{GELU}\Bigl(
\mathrm{Linear}\bigl(
\mathrm{LN}(&[\mathbf{r}^{(1)} \,\|\, \mathbf{r}^{(2)} \,\|\, \cdots \,\|\, \mathbf{r}^{(M)}])
\bigr)
\Bigr) \\
&\in \mathbb{R}^{D}.
\end{aligned}
\end{equation}
followed by the classifier
\begin{equation}
    \hat{\mathbf{y}} = \mathbf{W}_{\mathrm{cls}}\,\mathrm{LN}(\mathbf{z}) + \mathbf{b}_{\mathrm{cls}} \in \mathbb{R}^{K}.
\end{equation}

This final stage performs \emph{evidence aggregation rather than late averaging}. Each branch contributes a scale-specific sufficient statistic, and the fusion layer learns how to weigh transient and slow dynamics jointly for the target task.

\subsection{Implementation Details}
\label{sec:implementation}

\textbf{Model configuration.}
All experiments are conducted using carefully designed model configurations for each dataset, as summarized in Table~\ref{tab:model_config_all}. 
While most architectural hyperparameters are shared, such as the hidden dimension ($D = 128$), SSM expansion ($E = 2$), and FFN expansion ($E_{\text{ffn}} = 4$), the model depth ($N_{\text{layer}}$) is adjusted according to the characteristics of each dataset.

The multi-scale CNN patch embedding adopts three parallel branches with strides $(5, 10, 25)$ to capture temporal patterns at different resolutions. 
Regularization techniques, including feature dropout, stochastic depth, and channel dropout, are applied  across datasets to enhance robustness.

\textbf{Training protocol.}
We train models for a maximum of 50 epochs with a batch size of 512 using the AdamW optimizer~\cite{loshchilov2017decoupled} with an initial learning rate of $5 \times 10^{-4}$ and a weight decay of $0.1$ on datasets except APAVA. As for APAVA, we adopt a batch size of 128, since small datasets require more training steps.
The learning rate schedule consists of a linear warmup over the first 5 epochs (starting from $1\%$ of the peak learning rate) followed by cosine annealing~\cite{loshchilov2016sgdr} for the remaining epochs.
Gradient norms are clipped to a maximum of $4.0$ to stabilize training.
We choose the checkpoint with the best validation F1 for final evaluation on the test set.

\textbf{Loss function.} 
We adopt the standard cross-entropy loss as the training objective. 
To improve generalization, label smoothing with $\epsilon = 0.02$ is applied to the softmax predictions.

\textbf{Evaluation metrics.}
We report six evaluation metrics: accuracy, macro-averaged precision, macro-averaged recall, macro-averaged F1 score, and macro-averaged AUROC.
All experiments are conducted with five random seeds (41--45) in fixed training, validation, and test splits, and we report the mean and standard deviation in the runs.

\textbf{Hardware.}
All experiments are conducted on two NVIDIA RTX 3090 GPUs. Multi-GPU training is supported via PyTorch DistributedDataParallel (DDP) for the larger datasets.
\section{Experiments}
\label{sec:experiment}

We evaluated \MedMamba{} on six publicly available medical time-series benchmarks that span three distinct physiological modalities: electroencephalography (EEG), electrocardiography (ECG) and human activity. All experiments are conducted in a \emph{subject-independent} setting, where the subjects in the training, validation, and test sets are mutually exclusive, to faithfully simulate real-world clinical deployment in which the model must generalize to previously unseen patients.
\begin{table}[t]
\centering
\caption{Model configuration across different datasets.}
\label{tab:model_config_all}
\renewcommand{\arraystretch}{1.1}
\setlength{\tabcolsep}{5pt}
\resizebox{\linewidth}{!}{
\begin{tabular}{lcccccc}
\toprule
\textbf{Hyperparameter} 
& \textbf{PTB} 
& \textbf{PTB-XL} 
& \textbf{ADFTD} 
& \textbf{APAVA} 
& \textbf{UCI-HAR} 
& \textbf{SleepEDF} \\
\midrule
Hidden dim $D$              & 128 & 128 & 128 & 128 & 128 & 128 \\
Layers $N_{\text{layer}}$   & 3   & 3   & 4   & 4   & 4   & 4 \\
SSM expansion $E$           & 2   & 2   & 2   & 2   & 2   & 2 \\
FFN expansion $E_{\text{ffn}}$ & 4 & 4   & 4   & 4   & 4   & 4 \\
CNN strides                & (5,10,25) & (5,10,25) & (5,10,25) & (5,10,25) & (5,10,25) & (5,10,25) \\
Feature dropout $p_{\text{do}}$ & 0.1 & 0.1 & 0.1 & 0.1 & 0.1 & 0.1 \\
Stochastic depth $p_{\text{dp}}$ & 0.1 & 0.1 & 0.1 & 0.1 & 0.1 & 0.1 \\
Channel dropout $p_{\text{ch}}$ & 0.3 & 0.3 & 0.1 & 0.1 & 0.1 & 0.1 \\
\bottomrule
\end{tabular}
}
\end{table}
\subsection{Datasets}
\label{sec:datasets}

Table~\ref{tab:dataset_summary} summarizes the key statistics of the six datasets used in our experiments.

\begin{table}[htbp]
\centering
\caption{Summary of datasets used in this study.}
\label{tab:dataset_summary}
\renewcommand{\arraystretch}{1.1}
\setlength{\tabcolsep}{6pt}

\resizebox{0.49\textwidth}{!}{%
\begin{tabular}{lcccccc}
\toprule
\textbf{Dataset} & \textbf{\#Subject} & \textbf{\#Sample} & \textbf{\#Class} 
& \textbf{\#Channel} & \textbf{\#Timestamps} & \textbf{Modality} \\
\midrule
ADFTD     & 88    & 69,752   & 3  & 19 & 256    & EEG \\
PTB       & 198   & 64,356   & 2  & 15 & 300   & ECG \\
PTB-XL    & 17,596& 191,400  & 5  & 12 & 250   & ECG \\
APAVA     & 23    & 5,967    & 2  & 16 & 256    & EEG \\
UCI-HAR   & 30    & 10,299   & 6  & 9  & 128   & IMU (Acc+Gyro) \\
Sleep-EDF & 100   & 42,308   & 5  & 1  & 3000  & EEG \\
\bottomrule
\end{tabular}
}

\end{table}

\textbf{APAVA}~\cite{escudero2006analysis} is a public EEG dataset collected from 23 subjects (12 patients with Alzheimer's disease and 11 healthy controls) with 16 channels at a sampling rate of 256 \,Hz.
Each 5-second trial (1,280 timestamps) is first scaled with a standard scaler and then segmented into half-overlapping 1-second samples of 256 timestamps, yielding 5,967 samples in total.
For the subject-independent split, subjects with IDs \{15, 16, 19, 20\} and \{1, 2, 17, 18\} are assigned to the validation and test sets, respectively, and the remaining subjects constitute the training set.

\textbf{ADFTD}~\cite{miltiadous2023dice} is a public EEG dataset comprising recordings from 88 subjects across three diagnostic categories: 29 healthy controls, 23 frontotemporal dementia (FTD) patients, and 36 Alzheimer's disease (AD) patients.
The 19-channel signals are originally sampled at 500\,Hz with a 0.5--45\,Hz bandpass filter applied.
We downsample to 256\,Hz and segment each recording into non-overlapping 1-second samples of 256 timestamps, resulting in 69,752 samples.
For the subject-independent split, 60\%, 20\%, and 20\% of the subjects (stratified by class) are assigned to training, validation and test sets, respectively.

\textbf{PTB}~\cite{bank2000physiotoolkit} is a public ECG dataset recorded from 290 subjects with 15 channels at a raw sampling rate of 1,000\,Hz.
We use a subset of 198 subjects comprising myocardial infarction patients and healthy controls.
The signals are downsampled to 250\,Hz, normalized via standard scaling, and processed into single-heartbeat samples by detecting R-peaks across all channels.
Shorter heartbeats are zero-padded to a uniform length of 300 timestamps, producing 64,356 samples.
Subjects are divided into training, validation and test sets in a ratio of 55\%, 15\% and 30\% in a subject-independent manner.

\textbf{PTB-XL}~\cite{wagner2020ptb} is a large-scale public ECG dataset containing 12-lead recordings from 18,869 subjects, each annotated with one of five diagnostic categories: Normal, Myocardial Infarction (MI), ST/T Change (STTC), Conduction Disturbance (CD), and Hypertrophy.
After discarding subjects with inconsistent labels across trials, 17,596 subjects remain.
The 500\,Hz recordings are downsampled to 250\,Hz, normalized, and segmented into non-overlapping 1-second samples of 250 timestamps, yielding 191,400 samples.
For the subject-independent split, 60\%, 20\%, and 20\% of the subjects (stratified by class) are allocated to the training, validation, and test sets, respectively.

\textbf{Sleep-EDF}~\cite{bank2000physiotoolkit} is a large-scale public sleep-stage classification datasets. The objective of sleep stage classification is to categorize 30-second EEG signals into five distinct stages: Wake (W), Non-rapid eye movement (N1,N2,N3), and Rapid Eye Movement (REM). We make use of the Fpz-Cz channel, which captures EEG signals sampled at a rate of 100 Hz.  

\textbf{UCI-HAR}~\cite{anguita2013public} is a smartphone-based human activity recognition dataset collected from 30 subjects performing six daily activities. It contains nine-channel IMU signals (3-axis accelerometer, gyroscope, and body acceleration) sampled at 128Hz. The signals are segmented into 2.56-second windows (128 timestamps) with 50\% overlap, yielding 10,299 samples in total. We adopt the standard subject-independent split with 70\% of subjects for training and 30\% for testing.

All samples undergo per-channel z-score normalization before being fed into the model. Across all six datasets, samples from the same subject appear exclusively in one of the three splits, ensuring that no information leakage occurs between training and evaluation.

\subsection{Baselines}
To showcase the superiority of our methods, we extensively include the latest and advanced models in the community as basic baselines. These include general time series transformers: Autoformer~\cite{wu2021autoformer}, Crossformer~\cite{zhang2023crossformer}, FEDformer~\cite{zhou2022fedformer}, Informer~\cite{zhou2021informer}, MTST~\cite{zhang2024multi}, iTransformer~\cite{liu2023itransformer}, Nonformer~\cite{liu2022non}, PatchTST~\cite{nie2022time}, Reformer~\cite{kitaev2020reformer}, Transformer~\cite{vaswani2017attention}; specialized MedTS transformers: Medformer~\cite{wang2024medformer}.

\subsection{Centralization Analysis of Medical Time Series}

To better understand the structural characteristics of medical time series,
we quantify the degree of centralization using two complementary metrics:
spectral centralization index (SCI) and dynamic influence centralization (DIC).

\textbf{Spectral Centralization Index (SCI).}
Given a multivariate time series $\mathbf{X} \in \mathbb{R}^{S \times T}$,
where $S$ is the number of channels and $T$ is the sequence length,
we first remove the temporal mean:
\begin{equation}
\bar{\mathbf{X}} = \frac{1}{T}\mathbf{X}\mathbf{1}_T,
\end{equation}
and compute the covariance matrix:
\begin{equation}
\mathbf{\Sigma} = \frac{1}{T-1}(\mathbf{X}-\bar{\mathbf{X}})(\mathbf{X}-\bar{\mathbf{X}})^\top.
\end{equation}
The SCI is defined as:
\begin{equation}
\mathrm{SCI}(\mathbf{X}) = \frac{\lambda_{\max}(\mathbf{\Sigma})}{\mathrm{Tr}(\mathbf{\Sigma})},
\end{equation}
where $\lambda_{\max}(\cdot)$ denotes the largest eigenvalue.
A higher SCI indicates that the signal energy is concentrated in a
low-dimensional dominant subspace, reflecting strong spatial centralization.

\textbf{Dynamic Influence Centralization (DIC).}
To measure temporal dominance, we model the time series using a first-order
vector autoregressive (VAR) formulation. Let
\begin{equation}
\mathbf{Z} = [\mathbf{x}_1, \mathbf{x}_2, \dots, \mathbf{x}_{T-1}], \quad
\mathbf{Y} = [\mathbf{x}_2, \mathbf{x}_3, \dots, \mathbf{x}_T],
\end{equation}
where $\mathbf{x}_t$ is the $t$-th column of $\mathbf{X}$.
We estimate the transition matrix:
\begin{equation}
\mathbf{A} = \mathbf{Y}\mathbf{Z}^\top.
\end{equation}
The outgoing influence strength of channel $i$ is defined as:
\begin{equation}
s_i = \sum_{j} |A_{ji}|,
\end{equation}
and the DIC is computed as:
\begin{equation}
\mathrm{DIC}(\mathbf{X}) = \frac{\max_i s_i - \bar{s}}{\bar{s}}, \quad
\bar{s} = \frac{1}{S} \sum_{i} s_i.
\end{equation}
A higher DIC indicates that a small subset of channels dominates
the temporal dynamics of the system.

SCI measures spatial dominance as the concentration of signal energy
in the principal component of the covariance matrix, while DIC captures
temporal dominance as the imbalance of dynamic influence across channels.

Our empirical analysis as in Table \ref{tab:centralization} shows that physiological signals such as EEG and ECG
exhibit significantly higher centralization compared to general time series. This indicates that a small subset of channels
or physiological processes governs the global dynamics.

This observation has important implications for model design.
Most existing architectures, particularly Transformer-based methods,
implicitly assume uniform interactions across channels and time steps,
treating all tokens equally. However, such a fully distributed modeling
paradigm is mismatched with the inherently centralized structure of
medical time series.

In contrast, the design of \MedMamba{} aligns with this property.
The channel mixing module first aggregates cross-channel information,
allowing dominant channels to influence the global representation.
The bidirectional Mamba blocks further propagate these dominant dynamics
through selective state updates, while the multi-scale CNN patch embedding
captures centralization patterns across different temporal resolutions.
Together, these components form a centralized aggregating–redistributing
mechanism, which naturally fits the intrinsic structure of medical time
series and contributes to the observed performance gains.
\begin{table}[t]
\centering
\caption{Centralization metrics (SCI and DIC) comparing medical time series and climate data. 
Higher values indicate stronger centralized behavior.}
\label{tab:centralization}
\renewcommand{\arraystretch}{1.05}
\setlength{\tabcolsep}{4pt}

\begin{tabular}{lcccc|c}
\toprule
\textbf{Metric} 
& \multicolumn{4}{c|}{\textbf{Medical Time Series}} 
& \textbf{Climate} \\
\cmidrule(lr){2-5}
& ADFTD & APAVA & PTB & PTB-XL 
& Weather \\
\midrule
SCI & 0.92 & 0.52 & 0.62 & 0.65 & 0.38 \\
DIC & 0.67 & 0.73 & 0.83 & 0.78 & 0.34 \\
\bottomrule
\end{tabular}
\end{table}
\vspace{-10pt}
\subsection{Main Results}
\input{tables/main_results1}
\input{tables/main_results2}
\textbf{ECG Benchmarks.} 
Table~\ref{tab:jbhi_table2} reports the quantitative comparison on two widely-used ECG classification benchmarks, PTB (2-class) and PTB-XL (5-class). 
Our method consistently achieves the best performance across all evaluation metrics, demonstrating its strong capability in modeling medical time series.

\textbf{PTB (Binary Classification).} 
On the binary classification task, our model establishes a clear performance margin over all baselines. 
It achieves an accuracy of \textbf{85.97\%}, outperforming the strongest baseline, Medformer (83.50\%), by \textbf{+2.47\%}. 
Consistent improvements are also observed in precision, recall, F1-score, and AUROC, indicating superior discriminative ability and better class balance.

\textbf{PTB-XL (Multi-class Classification).} 
On the more challenging multi-class setting, our method continues to deliver consistent performance gains. 
It improves the accuracy from 73.30\% (Crossformer) to \textbf{74.39\%}, and achieves the highest AUROC of \textbf{91.08\%}. 
The improvements are stable across all evaluation metrics, suggesting strong generalization under increased task complexity.

\textbf{EEG and Human Activity Benchmarks.} 
To further evaluate the generalizability of our model across different modalities, we conduct experiments on EEG (ADFTD, APAVA,SleepEDF) and human activity (UCI-HAR) benchmarks, as summarized in Table~\ref{tab:jbhi_table1}. 
Our method consistently achieves the best or highly competitive performance across all datasets and metrics.

\textbf{EEG Classification (ADFTD, APAVA and SleepEDF).} 
On ADFTD, a challenging EEG classification benchmark, our model achieves the highest accuracy of \textbf{54.72\%} and the highest F1-Score performance(52.01\%), outperforming Medformer (53.27\% and 50.65\%) and other baselines.
This dataset is known for its high inter-subject variability and subtle pathological patterns, making it particularly difficult for existing methods.
Despite these challenges, our method consistently improves performance across precision, recall, F1-score and AUROC, demonstrating its effectiveness in capturing complex EEG dynamics.

On APAVA, the performance margin is more pronounced. 
Our method achieves \textbf{82.43\%} accuracy, significantly surpassing the strongest baseline (78.74\%) by nearly \textbf{+3.7\%}. 
Notably, consistent improvements across all evaluation metrics indicate enhanced robustness in handling noisy and subject-variant EEG signals.

In SleepEDF, where each input corresponds to a 30-second epoch that requires modeling long-range temporal dependencies, our method achieves the highest AUROC (\textbf{96.05\%}) and competitive performance in other metrics.
These results further validate the ability of our model to capture long-term temporal dynamics for sleep staging tasks.

\textbf{Human Activity Recognition (UCI-HAR).} 
On UCI-HAR, our model achieves the best performance across all metrics, reaching an accuracy of \textbf{96.16\%} and AUROC of  \textbf{99.83\%}.
This dataset contains structured multi-channel IMU signals with strong inter-channel correlations and well-defined temporal patterns.
The superior performance demonstrates the effectiveness of our model in capturing both cross-channel interactions and temporal dependencies in human activity data.

\textbf{Overall Analysis.} 
Across all datasets, covering ECG, EEG, and human activity signals, our method consistently outperforms a diverse set of baselines. 
This cross-domain superiority highlights its strong generalization capability and robustness.

We attribute these improvements to the ability of our model to effectively capture both short-term local patterns and long-range temporal dependencies, while maintaining stable representations across different signal modalities and task complexities.
% ================================================================

\subsection{Ablation Study}
We benchmark the effects of each component on the most challenging dataset ADFTD, to provide a comprehensive analysis of their individual contributions. Given the high inter-subject variability and subtle EEG patterns in ADFTD, this dataset serves as a rigorous testbed for evaluating the effectiveness and robustness of our design. The results are shown in Table~\ref{tab:ablation_vertical}.

\textbf{Effect of Multi-Scale Patch Embedding.}
Replacing the multi-scale design (strides $\{5, 10, 25\}$) with a single-scale of 5 reduces the F1-score from 52.01\% to 50.25\% ($-1.76\%$) and the AUROC from 73.53\% to 71.63\% ($-1.90\%$). This confirms that capturing temporal patterns at multiple resolutions is essential: the fine-grained scale (stride=5) captures fast transient features in the beta band, while the coarser scales (stride=10, 25) capture slower alpha- and theta-band dynamics that are characteristic of dementia-related EEG changes.

\textbf{Effect of Bidirectional Scanning.}
Replacing bidirectional scanning with a unidirectional (forward-only) variant leads to a consistent performance drop across all metrics. Specifically, the F1-score decreases from 52.01\% to 51.41\% ($-0.60\%$), while the AUROC drops from 73.53\% to 72.85\% ($-0.68\%$). Although the degradation is moderate compared to multi-scale modeling, the improvement is consistent across all evaluation metrics. This suggests that backward context provides complementary temporal cues that are not captured by forward dynamics alone, which is particularly beneficial for modeling symmetric or bidirectional temporal patterns commonly observed in physiological signals such as EEG oscillations.

\textbf{Effect of Channel Mixing.}
Removing the channel mixing module results in a more pronounced degradation, with the F1-score decreasing from 52.01\% to 50.98\% ($-1.03\%$) and the AUROC from 73.53\% to 72.21\% ($-1.32\%$). This indicates that explicitly modeling inter-channel dependencies plays a crucial role in extracting discriminative features from multi-channel physiological signals. Unlike temporal modeling components, which operate along the time axis, channel mixing captures cross-channel interactions that reflect underlying physiological coupling (e.g., between EEG electrodes), thereby enhancing representation quality.

\textbf{Summary.}
All three components contribute positively to the overall performance, with multi-scale modeling providing the largest gain, followed by channel mixing and bidirectional scanning. The full model consistently achieves the best results across all metrics, demonstrating that temporal multi-resolution modeling, bidirectional context integration, and cross-channel interaction are complementary and jointly necessary for effective representation learning in medical time series.

\input{tables/ablation}

\begin{figure}[!htbp]
    \centering
    \includegraphics[width=1.0\linewidth]{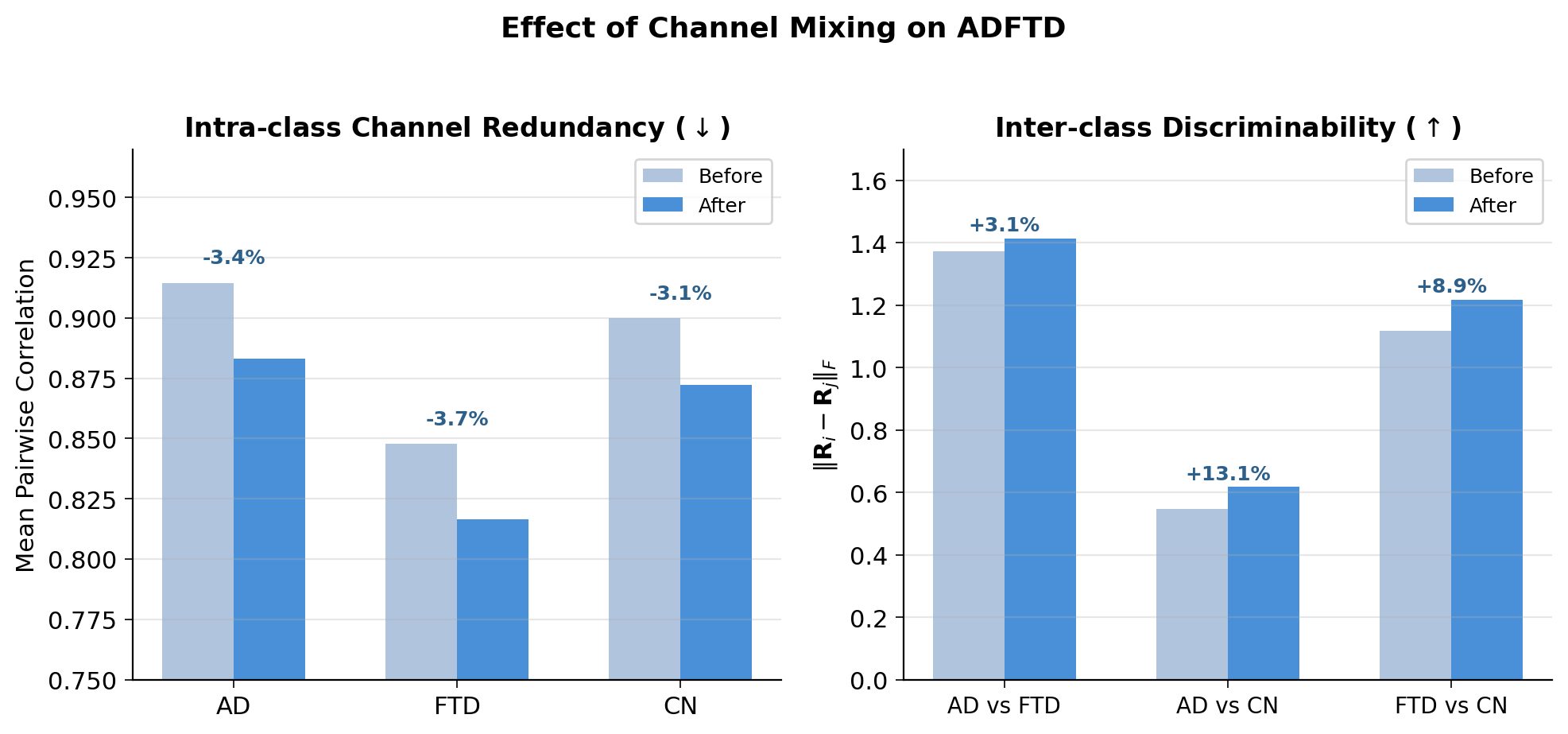}
    \caption{Effect of channel mixing on the ADFTD dataset. \textbf{Left}: mean pairwise Pearson correlation within each class: Alzheimer's Disease (AD), Frontotemporal Dementia (FTD), and Cognitively Normal (CN), where lower values indicate reduced channel redundancy. \textbf{Right}: inter-class discriminability measured by the Frobenius norm of correlation matrix differences, where higher values indicate greater separability. Percentages denote relative change after channel mixing.}
    \label{fig:channel_mix}
\end{figure}

\subsection{Effects of Channel Mixing}
\label{sec:channel_mix_analysis}
To understand how the channel mixing layer transforms the input representation, we analyze the channel correlation structure before and after the learned mixing on the ADFTD test set. Specifically, for each class $c$, we compute the average pairwise Pearson correlation matrix $\mathbf{R}_c \in \mathbb{R}^{C \times C}$ across all samples, and measure inter-class discriminability as the Frobenius norm of the pairwise difference $\|\mathbf{R}_i - \mathbf{R}_j\|_F$.

As shown in Fig.~\ref{fig:channel_mix}, the channel mixing layer produces two complementary effects. First, the mean intra-class channel correlation is consistently reduced across all three classes (AD: $-3.4\%$, FTD: $-3.7\%$, CN: $-3.1\%$), indicating that the layer learns to decorrelate redundant channel information. Second, the inter-class discriminability is simultaneously increased for all class pairs, with the most pronounced improvement observed for the clinically critical AD vs.\ CN comparison ($+13.1\%$), followed by FTD vs.\ CN ($+8.9\%$) and AD vs.\ FTD ($+3.1\%$).

These results suggest that the channel mixing layer acts as a lightweight spatial filter that suppresses shared channel redundancy while amplifying class-specific inter-channel patterns. Notably, the largest improvement occurs for AD vs.\ CN, where distinguishing pathological slowing from healthy brain activity relies heavily on spatial distribution differences across electrode sites. The modest number of additional parameters introduced by channel mixing (a two-layer MLP with expansion ratio $\rho = 2$ applied pointwise across channels) makes this an efficient mechanism for injecting spatial inductive bias without requiring complicated graph structure or electrode position information.
\subsection{Efficiency Analysis}
To evaluate the practical deployment potential of \MedMamba, we compare its computational efficiency against baseline methods in terms of parameter count and inference throughput. All models are benchmarked on the ADFTD dataset using a single NVIDIA RTX 3090 GPU with a batch size of 256.

The core efficiency advantage of \MedMamba{} stems from the linear computational complexity of the selective scan mechanism. While attention mechanisms in Transformer-based models incurs $\mathcal{O}(L^2 D)$ time complexity for a sequence of length $L$ and model dimension $D$, the SSM recurrence in Mamba requires only $\mathcal{O}(LDN)$, where $N$ is the state dimension ($N=16$ in our setting, $N \ll L$). Moreover, the multi-scale patch embedding further reduces the effective sequence length: at scale $m$ with stride $s_m$, the Mamba blocks operate on $L_m = \lfloor (L - s_m) / s_m \rfloor + 1$ tokens rather than the full input length, yielding a total complexity of $\sum_{m=1}^{M} \mathcal{O}(L_m D N)$. For instance, with strides $\{5, 10, 25\}$ on a 256-length input, the patched lengths are $\{51, 25, 10\}$, which are substantially shorter than $L$, further amplifying the efficiency gain.

Fig.~\ref{fig:efficiency} presents a scatter plot of F1-score versus inference speed (samples/sec), where the bubble size is proportional to the number of model parameters. \MedMamba{} achieves the best trade-off between accuracy and efficiency, positioned in the upper-right corner of the plot. Specifically, \MedMamba{} attains an F1-score of 52.0\% with only 5.7M parameters while processing 3.2k samples per second, approximately 2$\times$ faster than Transformer-based baselines (Transformer, Informer, Crossformer, FEDformer) operating at 1.4--1.6k samples/sec, and 4.5$\times$ faster than PatchTST and Autoformer (0.6--0.7k samples/sec). This empirical speedup directly reflects the linear-versus-quadratic complexity gap.

Compared to Medformer, which achieves a slightly lower F1-score of 50.8\%, \MedMamba{} uses 5.3$\times$ fewer parameters (5.7M vs.\ 30.0M) and runs 4.6$\times$ faster (3.2k vs.\ 0.7k samples/sec). The compact parameter footprint is attributable to the SSM parameterization, which encodes temporal dynamics through a low-dimensional state vector rather than full pairwise attention matrices.

Among models with comparable parameter budgets (4.7--5.8M), \MedMamba{} consistently outperforms all alternatives in both F1-score and throughput. These results demonstrate that the proposed architecture is not only accurate but also highly efficient, making it well-suited for real-time clinical applications such as continuous monitoring in intensive care units and edge deployment on wearable devices.

\begin{figure}
    \centering
    \includegraphics[width=1.0\linewidth]{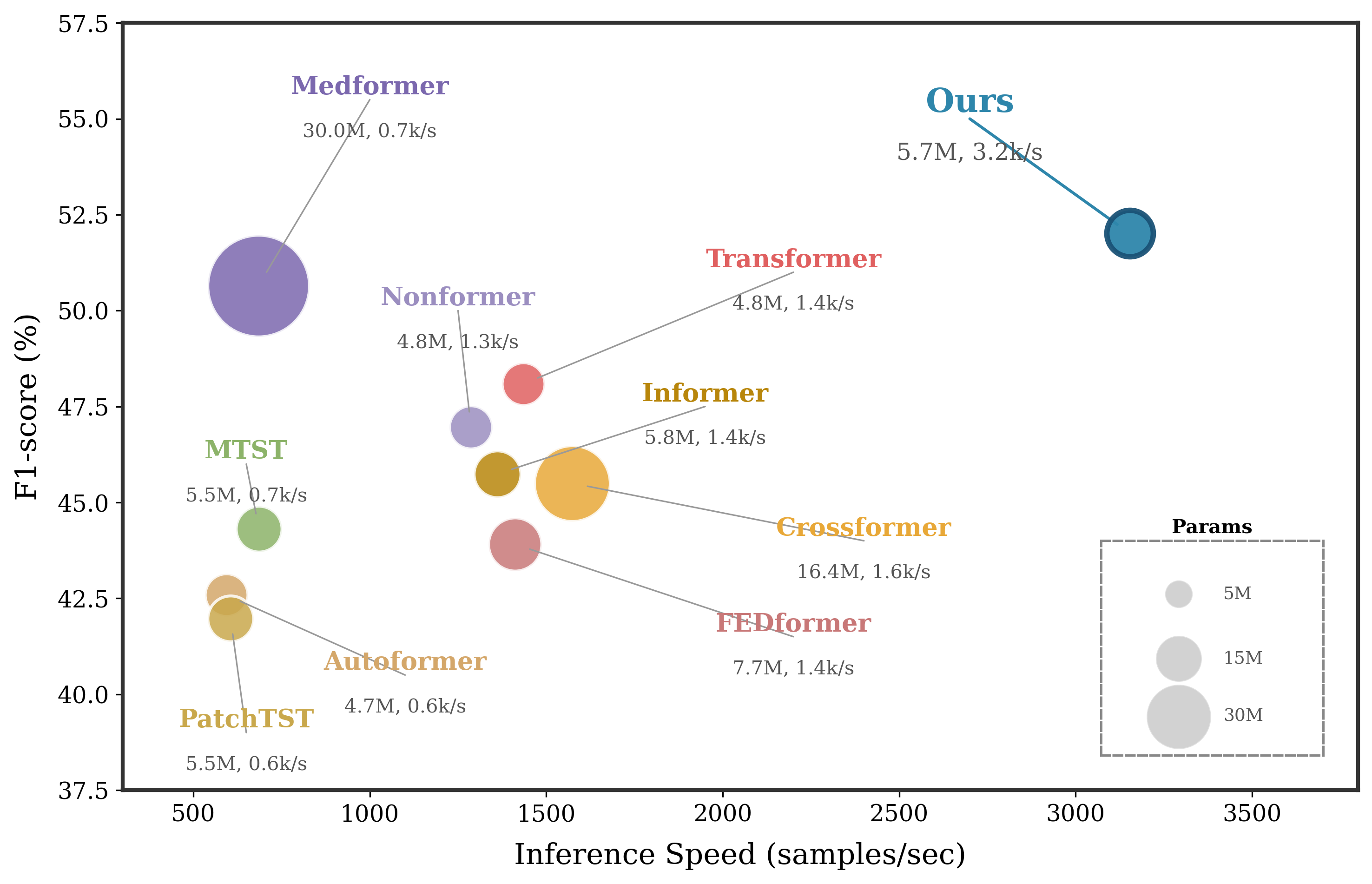}
    \caption{Accuracy--efficiency trade-off on the ADFTD dataset. Each bubble represents a model, with the horizontal axis indicating inference throughput (samples/sec), the vertical axis indicating macro F1-score (\%), and the bubble size proportional to the number of parameters. \MedMamba{} (upper-right) achieves competitive F1-score with the highest throughput and a compact parameter budget.}
    \label{fig:efficiency}
\end{figure}
\vspace{-14pt}

\subsection{t-SNE Visualization}
To showcase the quality of representations learned by \MedMamba , we apply t-SNE~\cite{van2008visualizing} to the representation features extracted from both \MedMamba{} and Medformer on the ADFTD and PTBXL test sets. As illustrated in Fig.~\ref{fig:tsne_comparison}, \MedMamba{} consistently yields more discriminative embeddings across both datasets.

On the three-class ADFTD task, the contrast is striking: \MedMamba{} produces well-separated clusters corresponding to CN, FTD, and AD, with compact intra-class distributions and clear inter-class margins. In contrast, Medformer features collapse into a largely homogeneous cloud with no discernible class structure, indicating that its encoder fails to capture the subtle spectral differences among dementia subtypes.

On the more challenging five-class PTBXL benchmark, \MedMamba{} continues to exhibit stronger class separation. Classes~0 and~1 form distinct, tightly grouped regions, and Class~3 clusters separately in the upper portion of the embedding space. Although Classes~2 and~4 partially overlap in both models, reflecting the inherent diagnostic similarity of certain cardiac conditions. \MedMamba{} still maintains noticeably tighter intra-class groupings compared to Medformer.

These visualizations corroborate the quantitative results in Table~\ref{tab:jbhi_table1} and Table~\ref{tab:jbhi_table2} and suggest that \MedMamba 's \space multi-scale bidirectional state-space modeling captures richer temporal patterns than Transformer-based attention, producing feature representations that are both more separable and more clinically meaningful.
\begin{figure}[!htbp]
    \centering
    \includegraphics[width=1.0\linewidth]{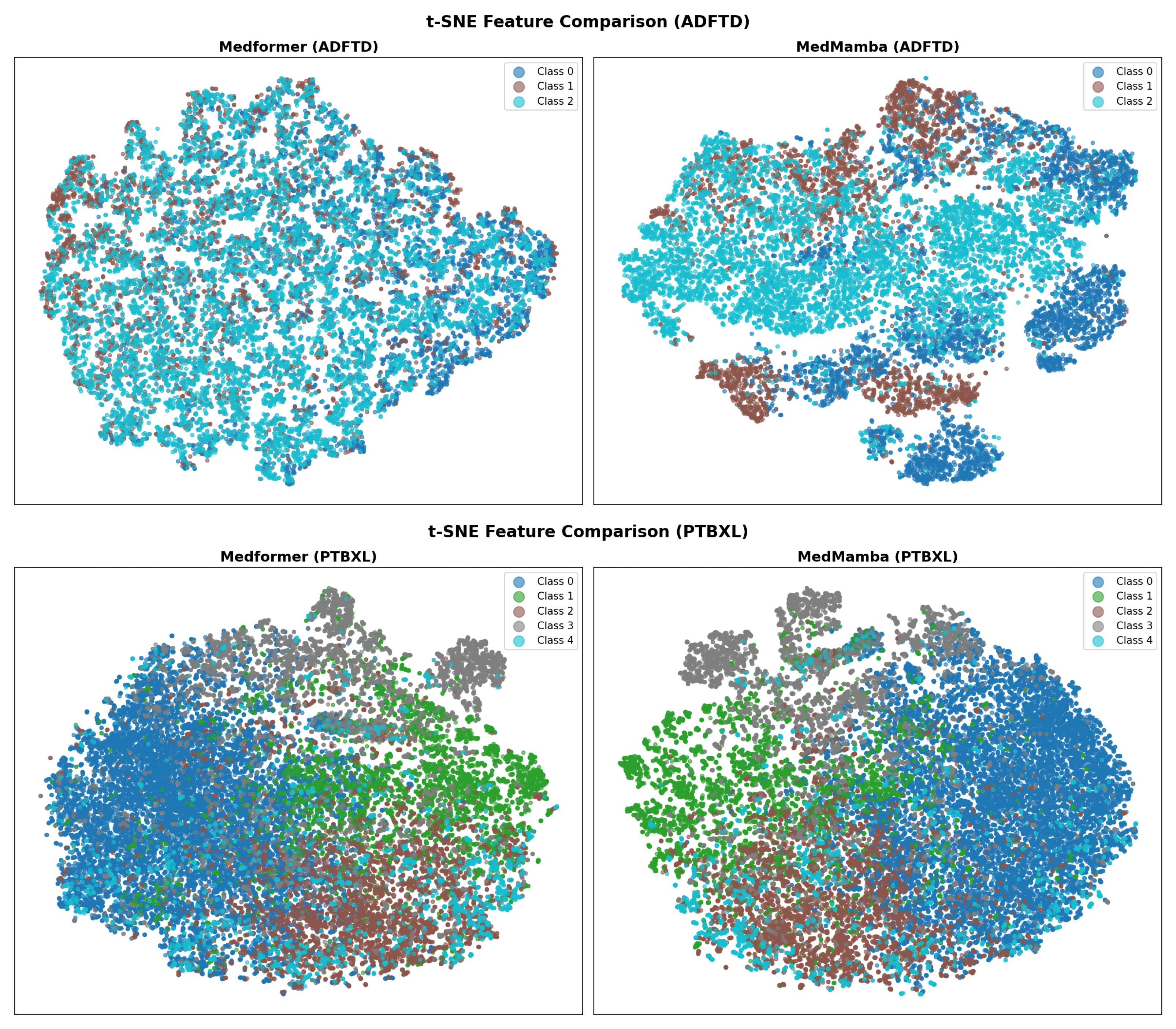}
    \caption{t-SNE visualization of learned feature representations on the ADFTD and PTBXL datasets. Each row corresponds to a dataset, with Medformer (left) and \MedMamba{} (right) shown side by side. On ADFTD (top), \MedMamba{} produces well-separated clusters for CN, FTD, and AD classes, whereas Medformer features remain largely entangled. On PTBXL (bottom), \MedMamba{} maintains clearer inter-class boundaries with more compact intra-class clusters, especially for Class~0 and Class~1, while Medformer exhibits substantial overlap across all five diagnostic categories. These results demonstrate that \MedMamba{} consistently learns more discriminative representations across datasets of varying complexity.}
    \label{fig:tsne_comparison}
\end{figure}
\section{Discussion and Future Directions}
\label{sec:discussion}

\subsection{Why Principle-Driven Design Matters for Medical Time Series}

Unlike generic time series, physiological signals exhibit strong structural regularities, including centralized spatial dependency, multi-timescale composition, and non-causal contextual patterns. 

Our results suggest that explicitly aligning model design with these properties is critical for achieving robust performance. In particular, the channel mixing module acts as a data-adaptive spatial filter, concentrating diagnostically relevant information into a lower-dimensional subspace, which is consistent with the high SCI and DIC values observed in medical signals. 

Similarly, multi-scale tokenization reduces the mismatch between model resolution and physiological event duration, enabling simultaneous modeling of transient patterns (e.g., QRS complexes or EEG spikes) and long-term rhythms (e.g., sleep stages). 

Finally, bidirectional state-space modeling provides a better inductive bias for offline diagnosis, where full temporal context is available, allowing the model to capture symmetric and phase-dependent patterns that are difficult to model with strictly causal architectures.
\subsection{From Architecture Design to Modeling Paradigm}

Beyond architectural improvements, our findings point to a broader modeling paradigm for medical time series. 

Most existing approaches treat physiological signals as generic sequences and rely on generic attention or convolution mechanisms. In contrast, our results suggest that incorporating domain-specific inductive biases is not merely beneficial but necessary. 

In particular, the combination of channel-first reparameterization, multi-rate temporal modeling, and bidirectional inference forms a structured pipeline that better matches the intrinsic organization of physiological signals. 

We believe this principle-driven design paradigm can serve as a general guideline for future models, extending beyond Mamba-based architectures to other sequence modeling frameworks.
\subsection{Implications for Clinical Deployment}

In addition to improved predictive performance, MedMamba offers practical advantages for real-world clinical deployment. 

The linear-time complexity and 4.6× inference speedup make it suitable for large-scale monitoring scenarios, such as wearable devices and sleep analysis systems. 
\subsection{Future directions.}
 Based on \MedMamba , several promising research directions emerge. First, self-supervised pre-training on large-scale unlabeled physiological recordings could improve representation quality. Second, incorporating cross-modal fusion (e.g., jointly modeling ECG and EEG) may capture complementary diagnostic information. Finally, deploying \MedMamba{} in prospective clinical studies with real-time inference constraints would provide the ultimate validation of its practical utility. 

\vspace{-5pt}
\section*{References}
\vspace{-18pt}
\bibliographystyle{IEEEtran}
\balance
\bibliography{JBHI}

\end{document}

%% file: tables/main_results1.tex
\begin{table*}[h]
\centering
\caption{Main results on ECG classification benchmarks}
\label{tab:jbhi_table2}
\renewcommand{\arraystretch}{1.08}
\setlength{\tabcolsep}{4pt}

\resizebox{\textwidth}{!}{%
\begin{tabular}{lccccc|ccccc}
\toprule
\multirow{2}{*}{\textbf{Method}} 
& \multicolumn{5}{c|}{\textbf{PTB (2 Classes)}} 
& \multicolumn{5}{c}{\textbf{PTB-XL (5 Classes)}} \\
\cmidrule(lr){2-6} \cmidrule(lr){7-11}
& \textbf{Accuracy} & \textbf{Precision} & \textbf{Recall} & \textbf{F1-score} & \textbf{AUROC}
& \textbf{Accuracy} & \textbf{Precision} & \textbf{Recall} & \textbf{F1-score} & \textbf{AUROC} \\
\midrule

\textbf{Autoformer}  
& 73.35$\pm$2.10 & 72.11$\pm$2.89 & 63.24$\pm$3.17 & 63.69$\pm$3.84 & 78.54$\pm$3.48 
& 61.68$\pm$2.72 & 51.60$\pm$1.64 & 49.10$\pm$1.52 & 48.85$\pm$2.27 & 82.04$\pm$1.44 \\

\textbf{Crossformer}  
& 80.17$\pm$3.79 & 85.04$\pm$1.83 & 71.25$\pm$6.29 & 72.75$\pm$7.19 & 88.55$\pm$3.45 
& 73.30$\pm$0.14 & 65.06$\pm$0.35 & 61.23$\pm$0.33 & 62.59$\pm$0.14 & 90.02$\pm$0.06 \\

\textbf{FEDformer}  
& 76.05$\pm$2.54 & 77.58$\pm$3.61 & 66.10$\pm$3.55 & 67.14$\pm$4.37 & 85.93$\pm$4.31 
& 57.20$\pm$9.47 & 52.38$\pm$6.09 & 49.04$\pm$7.26 & 47.89$\pm$8.44 & 82.13$\pm$4.17 \\

\textbf{Informer}  
& 78.69$\pm$1.68 & 82.87$\pm$1.02 & 69.19$\pm$2.90 & 70.84$\pm$3.47 & 92.09$\pm$0.53 
& 71.43$\pm$0.32 & 62.64$\pm$0.60 & 59.12$\pm$0.47 & 60.44$\pm$0.43 & 88.65$\pm$0.09 \\

\textbf{iTransformer}  
& 83.89$\pm$0.71 & \textbf{88.25$\pm$1.18} & 76.39$\pm$1.01 & 79.06$\pm$1.06 & 91.18$\pm$1.16 
& 69.28$\pm$0.22 & 59.59$\pm$0.45 & 54.62$\pm$0.18 & 56.20$\pm$0.19 & 86.71$\pm$0.10 \\

\textbf{MTST}  
& 76.59$\pm$1.90 & 79.88$\pm$1.90 & 66.31$\pm$2.95 & 67.38$\pm$3.71 & 86.86$\pm$2.75 
& 72.14$\pm$0.27 & 63.84$\pm$0.72 & 60.01$\pm$0.81 & 61.43$\pm$0.38 & 88.97$\pm$0.33 \\

\textbf{Nonformer}  
& 78.66$\pm$0.49 & 82.77$\pm$0.86 & 69.12$\pm$0.87 & 70.90$\pm$1.00 & 89.37$\pm$2.51 
& 70.56$\pm$0.55 & 61.57$\pm$0.66 & 57.75$\pm$0.72 & 59.10$\pm$0.66 & 88.32$\pm$0.36 \\

\textbf{PatchTST}  
& 74.74$\pm$1.62 & 76.94$\pm$1.51 & 63.89$\pm$2.71 & 64.36$\pm$3.38 & 88.79$\pm$0.91 
& 73.23$\pm$0.25 & 65.70$\pm$0.64 & 60.82$\pm$0.76 & 62.61$\pm$0.34 & 89.74$\pm$0.19 \\

\textbf{Reformer}  
& 77.96$\pm$2.13 & 81.72$\pm$1.61 & 68.20$\pm$3.35 & 69.65$\pm$3.88 & 91.13$\pm$0.74 
& 71.72$\pm$0.43 & 63.12$\pm$1.02 & 59.20$\pm$0.75 & 60.69$\pm$0.18 & 88.80$\pm$0.24 \\

\textbf{Transformer}  
& 77.37$\pm$1.02 & 81.84$\pm$0.66 & 67.14$\pm$1.80 & 68.47$\pm$2.19 & 90.08$\pm$1.76 
& 70.59$\pm$0.44 & 61.57$\pm$0.65 & 57.62$\pm$0.35 & 59.05$\pm$0.25 & 88.21$\pm$0.16 \\

\textbf{Medformer}  
& 83.50$\pm$2.01 & 85.19$\pm$0.94 & 77.11$\pm$3.39 & 79.18$\pm$3.31 & 92.81$\pm$1.48 
& 72.87$\pm$0.23 & 64.14$\pm$0.42 & 60.60$\pm$0.46 & 62.02$\pm$0.37 & 89.66$\pm$0.13 \\

\textbf{Ours}  
& \textbf{85.97$\pm$1.75} & 87.25$\pm$1.25 & \textbf{80.73$\pm$1.72} & \textbf{82.77$\pm$2.19} & \textbf{93.47$\pm$1.93} 
& \textbf{74.39$\pm$0.32} & \textbf{66.50$\pm$0.42} & \textbf{63.44$\pm$0.26} & \textbf{64.55$\pm$0.12} & \textbf{91.08$\pm$0.06} \\

\bottomrule
\end{tabular}%
}
\end{table*}

%% file: tables/main_results2.tex
\begin{table*}[h]
\centering
\caption{Main results on EEG and human activity benchmarks}
\label{tab:jbhi_table1}
\renewcommand{\arraystretch}{1.08}
\setlength{\tabcolsep}{4pt}

\resizebox{\textwidth}{!}{%
\begin{tabular}{cl*{12}{c}}
\toprule
\textbf{Dataset} & \textbf{Metric}
& Autoformer & Crossformer & FEDformer & Informer & iTransformer
& MTST & Nonformer & PatchTST & Reformer & Transformer & Medformer & Ours \\
\midrule

%%%%%%%%%%%%%%%%%%%%%%%%%%%%%%%%%%%%
% ADFTD
%%%%%%%%%%%%%%%%%%%%%%%%%%%%%%%%%%%%
\multirow{5}{*}{\shortstack{\textbf{ADFTD} \\ \textbf{(3 Classes)}}}
& Accuracy  
& 45.25$\pm$1.48 & 50.45$\pm$2.31 & 46.30$\pm$0.59 & 48.45$\pm$1.96 
& 52.60$\pm$1.59 & 45.60$\pm$2.03 & 49.95$\pm$1.05 & 44.37$\pm$0.95 
& 50.78$\pm$1.17 & 50.47$\pm$2.14 & 53.27$\pm$1.54
& \textbf{54.72$\pm$0.87} \\

& Precision 
& 43.67$\pm$1.94 & 45.57$\pm$1.63 & 46.05$\pm$0.76 & 46.54$\pm$1.68 
& 46.79$\pm$1.27 & 44.70$\pm$1.33 & 47.71$\pm$0.97 & 42.40$\pm$1.13 
& 49.64$\pm$1.49 & 49.13$\pm$1.83 & 51.02$\pm$1.57 
& \textbf{52.58$\pm$0.95} \\

& Recall    
& 42.96$\pm$2.03 & 45.88$\pm$1.82 & 44.22$\pm$1.38 & 46.06$\pm$1.84 
& 47.28$\pm$1.29 & 45.05$\pm$1.30 & 47.46$\pm$1.50 & 42.06$\pm$1.48 
& 49.89$\pm$1.67 & 48.01$\pm$1.53 & 50.71$\pm$1.55 
& \textbf{51.91$\pm$0.92} \\

& F1-score  
& 42.59$\pm$1.85 & 45.50$\pm$1.70 & 43.91$\pm$1.37 & 45.74$\pm$1.38 
& 46.79$\pm$1.13 & 44.31$\pm$1.74 & 46.96$\pm$1.35 & 41.97$\pm$1.37 
& 47.94$\pm$0.69 & 48.09$\pm$1.59 & 50.65$\pm$1.51 
& \textbf{52.01$\pm$0.93} \\

& AUROC     
& 61.02$\pm$1.82 & 66.45$\pm$2.03 & 62.62$\pm$1.75 & 65.87$\pm$1.27 
& 67.26$\pm$1.16 & 62.50$\pm$0.81 & 66.23$\pm$1.37 & 60.08$\pm$1.50 
& 69.17$\pm$1.58 & 67.93$\pm$1.59 & 70.93$\pm$1.19 
& \textbf{73.53$\pm$0.94} \\
\midrule

%%%%%%%%%%%%%%%%%%%%%%%%%%%%%%%%%%%%
% APAVA
%%%%%%%%%%%%%%%%%%%%%%%%%%%%%%%%%%%%
\multirow{5}{*}{\shortstack{\textbf{APAVA} \\ \textbf{(2 Classes)}}}
& Accuracy  
& 68.64$\pm$1.82 & 73.77$\pm$1.95 & 74.94$\pm$2.15 & 73.11$\pm$4.40 
& 74.55$\pm$1.66 & 71.14$\pm$1.59 & 71.89$\pm$3.81 & 67.03$\pm$1.65 
& 78.70$\pm$2.00 & 76.30$\pm$4.72 & 78.74$\pm$0.64
& \textbf{82.43$\pm$1.21} \\

& Precision 
& 68.48$\pm$2.10 & 79.29$\pm$4.36 & 74.59$\pm$1.50 & 75.17$\pm$6.06 
& 74.77$\pm$2.10 & 79.30$\pm$0.97 & 71.80$\pm$4.58 & 78.76$\pm$1.28 
& 82.50$\pm$3.95 & 77.64$\pm$5.95 & 81.11$\pm$0.84 
& \textbf{86.82$\pm$0.74} \\

& Recall  
& 68.77$\pm$2.27 & 68.86$\pm$1.70 & 73.56$\pm$3.55 & 69.17$\pm$4.56 
& 71.76$\pm$1.72 & 65.27$\pm$2.17 & 69.44$\pm$3.56 & 59.91$\pm$2.02 
& 75.00$\pm$1.61 & 73.09$\pm$5.01 & 75.40$\pm$0.66 
& \textbf{78.92$\pm$1.52} \\

& F1-score   
& 68.06$\pm$1.94 & 68.93$\pm$1.85 & 73.51$\pm$3.39 & 69.47$\pm$5.06 
& 72.30$\pm$1.79 & 67.01$\pm$1.32 & 69.74$\pm$3.84 & 55.97$\pm$3.05 
& 75.93$\pm$1.73 & 73.75$\pm$5.38 & 76.31$\pm$0.71 
& \textbf{80.17$\pm$1.58} \\

& AUROC  
& 75.94$\pm$3.61 & 82.39$\pm$3.52 & 83.72$\pm$1.97 & 70.46$\pm$4.91 
& 85.59$\pm$1.55 & 68.87$\pm$2.34 & 70.55$\pm$2.96 & 65.65$\pm$2.08 
& 73.94$\pm$1.40 & 73.23$\pm$2.67 & 83.20$\pm$0.91 
& \textbf{89.02$\pm$0.79} \\
\midrule

%%%%%%%%%%%%%%%%%%%%%%%%%%%%%%%%%%%%
% UCI-HAR
%%%%%%%%%%%%%%%%%%%%%%%%%%%%%%%%%%%%
\multirow{5}{*}{\shortstack{\textbf{UCI-HAR} \\ \textbf{(6 Classes)}}}
& Accuracy  & 82.38$\pm$2.31 & 89.74$\pm$1.08 & 90.16$\pm$0.81 & 90.30$\pm$0.36 & 84.30$\pm$1.15 & 89.79$\pm$0.31 & 90.01$\pm$0.47 & 87.11$\pm$1.28 & 88.44$\pm$2.02 & 88.86$\pm$1.65 & 91.65$\pm$0.74 & \textbf{96.16$\pm$0.43}\\
& Precision & 83.94$\pm$1.89 & 89.36$\pm$0.82 & 90.96$\pm$0.76 & 90.26$\pm$0.53 & 84.38$\pm$2.21 & 89.58$\pm$0.48 & 90.19$\pm$0.37 & 87.46$\pm$2.17 & 88.69$\pm$1.87 & 89.03$\pm$0.72 & 91.89$\pm$0.59 & \textbf{96.15$\pm$0.44}\\
& Recall    & 83.29$\pm$1.74 & 89.32$\pm$0.95 & 90.50$\pm$0.58 & 90.31$\pm$0.66 & 84.25$\pm$0.97 & 89.22$\pm$0.66 & 90.14$\pm$0.52 & 87.05$\pm$0.88 & 88.44$\pm$2.02 & 88.86$\pm$1.65 & 91.65$\pm$0.74 & \textbf{96.35$\pm$0.43}\\
& F1-score  & 80.82$\pm$2.04 & 89.70$\pm$1.10 & 90.43$\pm$1.02 & 90.21$\pm$0.79 & 84.28$\pm$0.76 & 89.31$\pm$0.29 & 90.87$\pm$0.19 & 88.88$\pm$1.26 & 88.41$\pm$1.98 & 88.80$\pm$1.67 & 91.61$\pm$0.75 & \textbf{96.17$\pm$0.43}\\
& AUROC     & 94.21$\pm$1.14 & 98.20$\pm$0.41 & 98.36$\pm$0.57 & 98.48$\pm$0.69 & 96.55$\pm$1.02 & 98.35$\pm$0.37 & 97.97$\pm$0.52 & 98.33$\pm$0.47 & 98.21$\pm$0.57 & 98.08$\pm$0.88 & 98.99$\pm$0.34 & \textbf{99.83$\pm$0.02}\\
\midrule

%%%%%%%%%%%%%%%%%%%%%%%%%%%%%%%%%%%%
% SleepEDF（重点修正）
%%%%%%%%%%%%%%%%%%%%%%%%%%%%%%%%%%%%
\multirow{5}{*}{\shortstack{\textbf{SleepEDF} \\ \textbf{(5 Classes)}}}
& Accuracy  & 68.78$\pm$1.23 & 80.76$\pm$0.91 & 69.16$\pm$1.14 & 73.60$\pm$1.02 & 81.68$\pm$0.74 & 82.40$\pm$0.83 & 75.51$\pm$0.96 & 81.35$\pm$0.87 & 72.99$\pm$1.18 & 70.30$\pm$1.27 & \textbf{84.05$\pm$0.76} & 83.33$\pm$0.88  \\
& Precision & 64.53$\pm$1.34 & 71.73$\pm$1.02 & 62.25$\pm$1.29 & 61.83$\pm$1.21 & 72.21$\pm$0.88 & 73.52$\pm$0.85 & 67.49$\pm$1.03 & 71.89$\pm$0.97 & 69.61$\pm$1.11 & 60.87$\pm$1.38 & 72.84$\pm$0.79 & \textbf{76.56$\pm$0.29}  \\
& Recall    & 57.20$\pm$1.41 & 72.27$\pm$1.05 & 58.16$\pm$1.33 & 61.59$\pm$1.26 & 72.83$\pm$0.86 & 73.30$\pm$0.89 & 64.81$\pm$1.08 & 70.36$\pm$0.99 & 59.21$\pm$1.30 & 61.14$\pm$1.24 & \textbf{73.42$\pm$0.82} & 72.42$\pm$1.00 \\
& F1-score  & 59.42$\pm$1.36 & 71.85$\pm$1.01 & 59.10$\pm$1.31 & 60.96$\pm$1.22 & 72.36$\pm$0.82 & 73.16$\pm$0.86 & 63.55$\pm$1.04 & 70.63$\pm$0.93 & 61.92$\pm$1.19 & 60.41$\pm$1.28 & 71.49$\pm$0.84 & \textbf{73.74$\pm$0.67} \\
& AUROC     & 88.38$\pm$0.94 & 94.21$\pm$0.63 & 87.82$\pm$0.97 & 89.14$\pm$0.88 & 93.81$\pm$0.54 & 94.56$\pm$0.51 & 90.57$\pm$0.79 & 94.42$\pm$0.58 & 92.03$\pm$0.71 & 87.82$\pm$0.92 & 95.23$\pm$0.52 & \textbf{96.05$\pm$0.13} \\
\bottomrule
\end{tabular}%
}
\end{table*}

%% file: tables/ablation.tex
\begin{table}[htbp]
\centering
\caption{Ablation study on the ADFTD dataset under the subject-independent setting.
\cmark\ and \xmark\ indicate whether a component is enabled or disabled, respectively.
MS, Bi, and CM denote Multi-Scale, Bidirectional, and Channel Mixing, respectively.}
\label{tab:ablation_vertical}
\renewcommand{\arraystretch}{1.1}
\setlength{\tabcolsep}{4pt}
\small
\begin{tabular}{c l c}
\toprule
\textbf{MS / Bi / CM} & \textbf{Metric} & \textbf{Value} \\
\midrule

\multirow{5}{*}{\shortstack{ \cmark\ /\ \cmark\ /\ \cmark}}
& Accuracy  & \textbf{54.72$\pm$0.87} \\
& F1        & \textbf{52.01$\pm$0.93} \\
& Precision & \textbf{52.58$\pm$0.95} \\
& Recall    & \textbf{51.91$\pm$0.92} \\
& AUROC     & \textbf{73.53$\pm$0.94} \\
\midrule

\multirow{5}{*}{\shortstack{ \xmark\ /\ \cmark\ /\ \cmark}}
& Accuracy  & 51.96$\pm$0.55 {\scriptsize($\downarrow$2.76)} \\
& F1        & 50.25$\pm$0.57 {\scriptsize($\downarrow$1.76)} \\
& Precision & 50.75$\pm$0.36 {\scriptsize($\downarrow$1.83)} \\
& Recall    & 50.20$\pm$0.65 {\scriptsize($\downarrow$1.71)} \\
& AUROC     & 71.63$\pm$0.48 {\scriptsize($\downarrow$1.90)} \\
\midrule

\multirow{5}{*}{\shortstack{ \cmark\ /\ \xmark\ /\ \cmark}}
& Accuracy  & 54.18$\pm$0.88 {\scriptsize($\downarrow$0.54)} \\
& F1        & 51.41$\pm$0.91 {\scriptsize($\downarrow$0.60)} \\
& Precision & 51.92$\pm$0.97 {\scriptsize($\downarrow$0.66)} \\
& Recall    & 51.28$\pm$0.95 {\scriptsize($\downarrow$0.63)} \\
& AUROC     & 72.85$\pm$0.90 {\scriptsize($\downarrow$0.68)} \\
\midrule

\multirow{5}{*}{\shortstack{ \cmark\ /\ \cmark\ /\ \xmark}}
& Accuracy  & 53.68$\pm$0.79 {\scriptsize($\downarrow$1.04)} \\
& F1        & 50.98$\pm$0.82 {\scriptsize($\downarrow$1.03)} \\
& Precision & 51.41$\pm$0.88 {\scriptsize($\downarrow$1.17)} \\
& Recall    & 50.87$\pm$0.80 {\scriptsize($\downarrow$1.04)} \\
& AUROC     & 72.21$\pm$0.70 {\scriptsize($\downarrow$1.32)} \\
\bottomrule
\end{tabular}
\end{table}